\documentclass[12pt]{article}

\usepackage[colorlinks]{hyperref}            
\usepackage{color}
\usepackage{graphicx,subfigure,amsmath,amssymb,amsfonts,bm,epsfig,epsf,url,dsfont}
\usepackage{fullpage}
\usepackage[small,bf]{caption}
\usepackage[top=1in,bottom=1in,left=1in,right=1in]{geometry}
\usepackage[sort&compress]{natbib} 
\usepackage{bbm}      
\usepackage{booktabs}
\usepackage{cases}
\usepackage{dsfont}
\usepackage{multirow}
\usepackage{graphicx}
\usepackage{tikz}
\usepackage{enumitem}
\usetikzlibrary{calc}
\usetikzlibrary{patterns}

\tikzstyle{every node}=[circle, draw, fill=black!50, inner sep=0pt, minimum width=4pt]
\tikzstyle{rouge}=[circle, draw, fill=red, inner sep=0pt, minimum width=6pt]
\tikzstyle{bleu}=[circle, draw, fill=blue, inner sep=0pt, minimum width=6pt]
\tikzstyle{petitrouge}=[circle, draw, fill=red, inner sep=0pt, minimum width=4pt]
\tikzstyle{petitbleu}=[circle, draw, fill=blue, inner sep=0pt, minimum width=4pt]
\tikzstyle{texte}=[draw=none, fill=none]
\tikzstyle{white}=[circle, draw, fill=black!0, inner sep=0pt, minimum width=4pt]

\definecolor{ffqqqq}{rgb}{1,0,0}
\definecolor{qqzzqq}{rgb}{0,0.6,0}
\definecolor{qqqqff}{rgb}{0,0,1}
\definecolor{uuuuuu}{rgb}{0.27,0.27,0.27}

\newtheorem{remark}{Remark}
\newtheorem{theorem}{Theorem}

\newtheorem{lemma}{Lemma}

\newenvironment{proof}{\medskip\noindent{\bf Proof.}}{\hfill$\Box$\vspace*{1mm}\medskip}

\newcommand{\ut}{\underline{t}}
\newcommand{\p}{\mathbf{p}}
\newcommand{\q}{\mathbf{q}}
\newcommand{\PP}{\mathcal{P}}
\newcommand{\QQ}{\mathcal{Q}}
\newcommand{\eps}{\varepsilon}

\renewcommand{\phi}{\varphi}

\renewcommand{\P}{\mathbb{P}}
\newcommand{\E}{\mathbb{E}}
\newcommand{\N}{\mathbb{N}}

\newcommand{\cA}{\mathcal{A}}

\newcommand{\cB}{\mathcal{B}}

\newcommand{\cP}{\mathcal{P}}

\newcommand{\cD}{\mathcal{D}}

\def\ds1{\mathds{1}}
\renewcommand{\epsilon}{\varepsilon}

\renewcommand{\tilde}{\widetilde}

\newlength{\minipagewidth}
\newcommand{\beq}{\begin{equation}}
\newcommand{\eeq}{\end{equation}}

\newcommand{\beqa}{\begin{eqnarray}}
\newcommand{\eeqa}{\end{eqnarray}}

\newcommand{\beqan}{\begin{eqnarray*}}
\newcommand{\eeqan}{\end{eqnarray*}}

\def\ba#1\ea{\begin{align*}#1\end{align*}} 
\def\banum#1\eanum{\begin{align}#1\end{align}} 

\begin{document}
\title{Coordination without communication: \\
optimal regret in two players multi-armed bandits}
\author{S\'ebastien Bubeck \\
Microsoft Research
\and Thomas Budzinski \thanks{This work was done while T. Budzinski was visiting Microsoft Research.} \\
University of British Columbia}
\maketitle

\begin{abstract}
We consider two agents playing simultaneously the same stochastic three-armed bandit problem. The two agents are cooperating but they cannot communicate. We propose a strategy with no collisions at all between the players (with very high probability), and with near-optimal regret $O(\sqrt{T \log(T)})$. We also argue that the extra logarithmic term $\sqrt{\log(T)}$ should be necessary by proving a lower bound for a full information variant of the problem.
\end{abstract}

\section{Introduction} \label{sec:intro}
We consider the (cooperative) multi-player version of the classical stochastic multi-armed bandit problem. We focus on the case of two players, Alice and Bob, and three actions. The problem can be defined as follows. The environment\footnote{We focus on $\{0,1\}$-valued losses. Note that it is easy to reduce $[0,1]$-valued losses to $\{0,1\}$.} is described by the mean losses $\p = (p_1, p_2, p_3) \in [0,1]^3$ for the three actions. The parameter $\p$ is unknown to the players. Denote $(\ell_t(i))_{1 \leq i \leq 3, 1 \leq t \leq T}$ for a sequence of independent random variables such that $\P(\ell_t(i) = 1) = p_i$ and $\P(\ell_t(i) = 0) = 1-p_i$. At each time step $t=1, \hdots, T$, Alice and Bob choose independently two actions $i_t^A \in \{1,2,3\}$ and $i_t^B \in \{1,2,3\}$. If they collide, i.e. $i_t^A = i_t^B$, then they both suffer the maximal loss of $1$. Otherwise they respectively suffer the losses $\ell_t(i_t^A)$ and $\ell_t(i_t^B)$. As is usual in bandit scenarios, each player receives only its own loss as feedback (in particular when a player receives a loss of $1$, they don't know if they have collided or if it came from the loss $\ell_t$). The goal of the players is to minimize their (combined) cumulative losses. To evaluate the performance of Alice and Bob we measure the regret $R_T$, defined as the (expected) difference between their cumulative losses and the best they could have done if they knew $\p$, namely $T \cdot \p^*$ where $\p^* = \min(p_1 + p_2, p_1 + p_3, p_2 + p_3)$. That is:
\begin{equation} \label{eq:regret}
R_T = \sum_{t=1}^T \bigg( 2 \cdot \mathbbm{1}_{i_t^A = i_t^B} + \mathbbm{1}_{i_t^A \neq i_t^B} (p_{i_t^A} + p_{i_t^B}) - \p^* \bigg) \,.
\end{equation}

\subsection{Main result and related works}
The above problem is motivated by cognitive radio applications, where players correspond to devices trying to communicate with a cell tower, and the actions correspond to different channels. The model was first introduced roughly at the same time in \citet{LJP08, LZ10, AMTS11}, and has been extensively studied since then \citep{AM14, RSS16, BBMKP17, LM18, BP18, ALK19, BLPS19}. Despite all this attention, at the moment the state of the art regret bound is $\tilde{O}(T^{3/4})$. The latter regret was obtained for two players in \cite{BLPS19} (in fact it holds in the more general non-stochastic case), and it can also be recovered from the bounds in \cite{LM18, BP18} as we explain in the end of Section \ref{sec:warmup}. On the other hand no non-trivial lower bound is known (i.e. only $\Omega(\sqrt{T})$ is known). A near-optimal regret of $\tilde{O}(\sqrt{T})$ has been obtained under various extra assumptions such as revealed collisions, or assuming that players can abstain from playing, or assuming that the mean losses are bounded away from $1$ \citep{LM18, BP18, BLPS19}.

Our main contribution is the first $\tilde{O}(\sqrt{T})$ algorithm for this problem, in the case where there are $3$ arms (Theorem \ref{thm:main} below assumes shared randomness between the players, and Theorem \ref{thm:add} gives a deterministic strategy with $\tilde{O}(\sqrt{T})$ regret):


\begin{theorem} \label{thm:main}
There exists a randomized strategy (with shared randomness) for Alice and Bob such that, for any $\p \in [0,1]^3$, we simultaneously have
\[
\E [R_T] \leq 2^{20} \sqrt{T \log(T)} \,
\]
and
\begin{equation} \label{eq:nocollisions}
\P \left( \forall t \in [T], i_t^A \neq i_t^B \right) \geq 1-\frac{1}{T} \,,
\end{equation}
where the expectation and the probability are with respect to both the loss sequence and the randomness in Alice and Bob's strategies\footnote{By our method, we can actually obtain a slightly stronger version where, with probability at least $1-1/T$ with respect to the i.i.d. loss sequence, we have both the expected regret bound and almost surely no collision (with respect to the players' randomness).}.
\end{theorem}

The property \eqref{eq:nocollisions} is an important part of our result, and it points to a fundamental difference between our approach and all previous works on cooperative multi-player multi-armed bandits. Indeed, all previous works have proposed strategies that use collisions as a form of implicit communication between the players, since Alice can affect Bob's feedback by trying to force collisions. For example, assume as in \cite{LM18, BP18} that the mean-losses are bounded from above by $1-\mu$, i.e., $\|\p\|_{\infty} \leq 1-\mu$. Then if Bob plays an action for $\Omega(1/\mu)$ rounds and does not observe a single $0$ loss, he knows that with high probability Alice must have been playing that action too, effectively making communication possible. Leveraging this implicit communication device, \cite{LM18, BP18} obtain a strategy with regret $\tilde{O}(\sqrt{T} + 1/\mu)$ (we explain at the end of Section \ref{sec:warmup} how to use this result to obtain an algorithm with $\tilde{O}(T^{3/4})$ regret without any assumption). In \cite{BLPS19} another $\tilde{O}(T^{3/4})$ strategy is proposed. It is epoch-based, with Alice playing a fixed action in an epoch, and Bob playing a sleeping-bandit strategy where arms awaken as losses with value $0$ are observed (i.e., an arm is awake for Bob when he can guarantee that Alice is not there for this epoch). Thus we see that both methods heavily rely on collisions for implicit communication. The approach presented in this paper is fundamentally different, in that with very high probability the two players {\em do not collide at all}. Thus we achieve one of the key properties required by the underlying cognitive radio application, namely that the two agents {\em do not communicate in any way} once the game has started. 

\subsection{Replacing shared randomness by few collisions}
The strategy we build to prove Theorem \ref{thm:main} crucially relies on having shared randomness for Alice and Bob. We do not know whether this assumption can be relaxed, while still maintaining both $O(\sqrt{T \log(T)})$ regret and the no-collision property \eqref{eq:nocollisions}. We show however that if one is willing to give up on the no-collision property, and exploit the ``implicit communication" allowed by the extra losses of $1$ due to collisions, then one can in fact obtain $O(\sqrt{T \log(T)})$-regret with deterministic strategies:

\begin{theorem} \label{thm:add}
There exists a deterministic strategy for Alice and Bob such that, with probability at least $1- \Omega(1/T)$,
\[
R_T = O(\sqrt{T \log(T)}) \,.
\]
\end{theorem}

We also note that for a toy variant of the problem (described next) we do give a $\tilde{O}(\sqrt{T})$-regret no-collision strategy without shared randomness. This result is based on a certain derandomization technique which seems hard to apply in the case of Theorem \ref{thm:main}.
%

We prove Theorem \ref{thm:add} in Section \ref{sec:add}.

\subsection{A toy problem}
In order to motivate our new strategy with no collisions (Theorem \ref{thm:main}), it will be useful to first consider a different model which contains the essence of the difficulty of {\em coordination without communication}, but without the usual {\em exploration or exploitation} dilemma. The first modification that we propose is to assume that, even under collisions, a ``real'' loss is revealed. Precisely, if both players play the same action $i$ at round $t$, then we assume that they both observe independent samples from $\mathrm{Ber}(p_i)$ (rather than observing $1$ in the original model). This modification completely removes the possibility for implicit communication, since Alice's feedback is now completely unaffected by the presence of Bob (and vice versa). Concretely we denote
$(\ell_t^X(i))_{1 \leq i \leq 3, 1 \leq t \leq T, X \in \{A,B\}}$ for a sequence of independent random variables such that $\P(\ell_t^X(i) = 1) = p_i$ and $\P(\ell_t^X(i) = 0) = 1-p_i$. When player $X \in \{A,B\}$ plays action $i$, they observe the loss $\ell_t^X(i)$ (irrespective of the other player's action). Note that in this model we still assume that the players suffer a loss of $1$ if they collide, they simply don't observe their actual suffered loss (to put it differently, we are still concerned with the regret \eqref{eq:regret}). The problem now looks significantly more difficult for the players\footnote{It is not strictly speaking more difficult, since always receiving the feedback $\ell^X_t(i_t)$ means that the players have a slightly more accurate estimate of $\p$.}, and it is not clear a priori that any non-trivial guarantee can be obtained. In fact it is non-trivial even with {\em full information}: that is at the end of round $t$, player $X \in \{A,B\}$ observes $(\ell_t^X(1), \ell_t^X(2), \ell_t^X(3))$. For this modified model we assume such a full information feedback. The reason why we have chosen to have two different, independent loss sequences $\ell^A$ and $\ell^B$ is that if we had $\ell^A=\ell^B$, then $A$ and $B$ would have exactly the same information, in which case it is very easy to avoid collisions.
\newline

Our first task will be to give a strategy with regret $O(\sqrt{T \log(T)})$ for the full-information toy model, which we do in Section \ref{sec:toyupper}. The extension to the bandit scenario is then done in Section \ref{sec:banditupper}. An interesting property of the toy model is that it is amenable to lower bound arguments, since we avoid the difficulty created by implicit communication. In particular we prove the first non-trivial lower bound for multi-player online learning, by showing that the extra factor $\sqrt{\log(T)}$ is necessary:

\begin{theorem} \label{thm:lower}
There exists a universal constant $c>0$ and a distribution over $\p$ such that, for any strategy in the full-information toy model, one has:
\[
\E_{\p} R_T \geq c \sqrt{T \log(T)} \,.
\]
\end{theorem}
Unfortunately, there does not seem to be a direct way to transfer this lower bound to the original bandit problem.
%

\section{Difficulties of coordination without communication} \label{sec:warmup}
Whether we consider the toy model, or strategies for the bandit scenario that do not exploit the extra $1$'s due to collisions, we face the same question: how can two agents with imperfect information coordinate without communicating? 
In this section we illustrate some of the difficulties of {\em coordination without communication}. We focus on the most basic bandit strategy, namely explore then exploit. We show how to appropriately modify it to obtain $T^{4/5}$ regret for the bandit scenario, using shared randomness. All the discussion applies similarly to the full-information toy model, and as we note at the end of the section it gives $T^{3/4}$ regret in that case.

\subsection{Explore then exploit}
Consider the following protocol:
\begin{enumerate}
\item Alice and Bob first explore in a round-robin way for $\Theta(T^b)$ rounds, where $b \in (0,1)$ is a fixed parameter. Denote $q^A(i)$ for the average loss observed by Alice on action $i$ (and similarly $q^B(i)$ for Bob).
\item Using these estimates, the players can order the arms in terms of expected performances. Denote $(A_1, A_2, A_3)$ (respectively $(B_1, B_2, B_3)$) for the order Alice (respectively Bob) obtains, in ascending order of average empirical loss (i.e., $q^A(A_1) \leq q^A(A_2) \leq q^A(A_3)$).
\item For the remaining rounds they want to exploit. Alice and Bob could have agreed that Alice will play the best action, and Bob the second best, thus for the remaining of the game Alice plays $A_1$ and Bob plays $B_2$.
\end{enumerate}
The problem with this naive implementation of explore/exploit is clear: there could be ambiguity on which action is the best, for example if $p_1 = p_2 \ll p_3$, in which case both $A_1$ and $B_2$ are independent and uniform in $\{1,2\}$.
Thus in this case there is a constant probability of collision, resulting in a linear regret. A natural fix is for Alice to build a set of ``potential top action'' $\cA$ and for Bob to build a set of ``potential second  best action'' $\cB$. To decide whether an action is ``potentially the top action'' we fix an ``ambiguity threshold'' $\tau$, and now replace step 3 above with:
\begin{enumerate}
\item[3'] If $q^A(A_1) \leq q^A(A_2) - \tau$ ($A_1$ is ``clearly'' the best) then let $\cA = \{A_1\}$ (in the same case for $B$ let $\cB = B_2$), if not but $q^A(A_2) \leq q^A(A_3) - \tau$ ($A_3$ is ``clearly'' worse than $A_1$ and $A_2$) then let $\cA=\{A_1, A_2\}$ (in the same case for $B$ let $\cB = \{B_1,B_2\}$), and if neither then let $\cA =\{1,2,3\}$ (same for $B$). To avoid collisions it makes sense for Alice to play $\min(\cA)$ and for Bob to play $\max(\cB)$.
\end{enumerate}
Unfortunately this is just pushing the problem to a different configuration of $\p$. Indeed consider for example $p_3 \gg p_1 > p_2 = p_1 - \tau$. With a constant probability Alice could end up with $\cA = \{2\}$ and Bob with $\cB = \{1,2\}$, in which case we have again a collision, and hence we get linear regret.

\subsection{The root of the problem}
Geometrically, the issues above come from the boundary regions of the ``decision map'' $\sigma : ([0,1]^3)^2 \rightarrow \{1,2,3\}^2$ from empirical estimates of the mean-losses to actions to be played in the exploitation phase. All our results will come from careful considerations of these boundaries. Moreover, most of the difficulties already arise for our proposed full-information toy model, hence the focus on the toy model first. We also note that the geometric considerations are much easier with two players and three actions, which is why we focus on this case in this paper. The ``high-dimensional'' version of the strategy proposed in Section \ref{sec:toyupper} probably requires different tools.
\newline

Before going into the geometric considerations, we can illustrate one of our insights in the simple case of the explore/exploit strategy above. Namely we propose to make the decision boundaries {\em random}. For the explore/exploit strategy this means taking the ambiguity threshold $\tau$ to be random. Say we take it random at scale $T^{-a}$ for some parameter $a \in [0,1]$. More precisely let $\tau = U / T^a$ with $U$ a uniform random variable in $[0,1]$. In particular, since we don't distinguish differences below the scale $T^{-a}$, we might suffer a regret of $T^{1-a}$. On the other hand, the only risk of collision is if Alice and Bob disagree on whether some gap $\Delta = |p_i - p_j|$ is smaller than $\tau$ or not. Since the fluctuations of the empirical means are of order $T^{-b/2}$, we have that a collision might happen if $|\tau - \Delta| = \tilde{O}(T^{-b/2})$. To put it differently, with high probability (over the observed losses during the exploration phase), collisions happen only if
\[
|U - T^{a} \Delta| = \tilde{O}(T^{a - b/2}) \,.
\]
Because we have taken $U$ uniform on $[0,1]$, the above event has probability (over the realization of $U$) at most $\tilde{O}(T^{a - b/2})$. Thus finally we get a regret of order:
\[
T \cdot T^{a - b/2} + T \cdot T^{-a} + T^b \,,
\]
which is optimized at $b=4/5$ and $a=1/5$, resulting in a $\tilde{O}(T^{4/5})$ regret. 

\subsection{Minor variants}
We note that the same argument applies to the full-information toy model, where we are effectively taking $b=1$, resulting in a $\tilde{O}(T^{3/4})$ regret. Furthermore the same technique can be used to estimate $\mu$ in \cite{LM18, BP18}, improving upon the above $T^{4/5}$ to give $T^{3/4}$ for the bandit case.

\section{Toy model upper bound} \label{sec:toyupper}
We prove here the following theorem:
\begin{theorem} \label{thm:toyupper}
There exists a deterministic strategy for Alice and Bob in the full-information toy model such that with probability at least $1-1/T$, one has both:
\begin{equation} \label{eq:toyregret}
R_T \leq 320 \sqrt{T \log(T)} \,,
\end{equation}
and $\forall t \in [T], i_t^A \neq i_t^B$.
\end{theorem}
For $2 \leq t \leq T$, $i \in \{1,2,3\}$ and $X \in \{A,B\}$, we write
\[q_t^X(i)=\frac{1}{t-1} \sum_{s=1}^{t-1} \ell_s^X(i), \]
with the convention $q_1^X(i)=0$. In other words $q_t^X$ is the estimate of the vector $\p$ by player $X$ at time $t$. Our strategy is based on a subtle partition of the cube $[0,1]^3$. Precisely we build a map $\sigma_t : [0,1]^3 \rightarrow \{1,2,3\} \times \{1,2,3\}$, with $\sigma_t=\left( \sigma_t^A, \sigma_t^B \right)$, such that Alice plays $i_t^A = \sigma^A_{t}(q_t^A)$ and Bob plays $i_t^B = \sigma_{t}^B(q_t^B)$. An interesting aspect of Theorem \ref{thm:toyupper} compared to Theorem \ref{thm:main} is that we do not require shared randomness for the full-information toy model. However it will be easier for us to first describe a shared randomness strategy, and then explain how to remove that assumption. More precisely, we first build a {\em random partition} $\sigma$, and we prove Theorem \ref{thm:toyupper} with \eqref{eq:toyregret} holding in expectation over this random partition. We explain how to derandomize in Section \ref{sec:dynamic} with a {\em dynamic} partition.
\newline

We denote $w_t = 16 \sqrt{\frac{\log(T)}{t}}$, and we fix the event 
\begin{equation} \label{eq:Omega}
\Omega = \left\{\forall t \in [T], i \in \{1,2,3\}, X\in\{A,B\}, |q_t^X(i) - p_i| < \frac{w_t}{4} \right\} \,.
\end{equation}
Applying Hoeffding's inequality and an union bound, one obtains
\[
\P(\Omega) \geq 1 - \frac{1}{T} \,.
\]
For the remainder of the section, we fix loss sequences for which $\Omega$ holds true. All probabilities will be taken with respect to the randomness of Alice and Bob. We note in particular that under $\Omega$ we have $\|q_t^X - p \|_{\infty} \leq \frac{w_t}{4}$ for $X \in \{A,B\}$, so we get 
\begin{equation} \label{eq:distance}
\|q_t^A - q_t^B\|_2 < w_t \,.
\end{equation}

\subsection{A random partition of the cube}

\subsubsection{Cylindrical coordinates} \label{sec:cylindrical}
To describe our partition, it will be more convenient to use cylindrical coordinates around the axis $\mathcal{D}=\{ \p | p_1=p_2=p_3 \}$. More precisely, for $\p=(p_1,p_2,p_3)$ we write
\[ 
m_{\p}=\frac{p_1+p_2+p_3}{3} \,,
\]
\[ 
r_{\p}=d(\p, \mathcal{D})=\sqrt{(p_1-m_{\p})^2+(p_2-m_{\p})^2+(p_3-m_{\p})^2} \,,
\]
and $\theta_{\p} \in [0,2\pi)$ for the angle between the line from $\p$ to its orthogonal projection $\left( m_{\p}, m_{\p}, m_{\p} \right)$ on the axis $\mathcal{D}$ and the half-line $\left\{ \left( m_{\p}-t, m_{\p}+2t, m_{\p}-t \right) | t \geq 0 \right\}$ (this angle is contained in the plane orthogonal to $\mathcal{D}$ passing through $\p$). 
We write $\p=(p_1,p_2,p_3)=[m_{\p}, r_{\p}, \theta_{\p}]$.
\newline

An equivalent way to describe these cylindrical coordinates is as follows. Let us denote $\mathbf{a} = \frac{1}{\sqrt{3}} (1,1,1)$ (the main axis direction), $\mathbf{b} = \sqrt{\frac{2}{3}} \left( - \frac12, 1, - \frac12 \right)$ (the direction of the half-line mentioned above), and $\mathbf{c} =  \sqrt{\frac{2}{3}} \left(\frac{\sqrt{3}}{2}, 0, - \frac{\sqrt{3}}{2} \right)$ (the direction so that $\{\mathbf{a},\mathbf{b},\mathbf{c}\}$ forms an orthonormal basis).
We have:
\begin{eqnarray*}
\p & = & \langle \p, \mathbf{a} \rangle \mathbf{a} + r_{\p} \cos(\theta_{\p}) \mathbf{b} + r_{\p} \sin(\theta_{\p}) \mathbf{c} \\
& = & \begin{pmatrix} m_{\p} \\ m_{\p} \\ m_{\p} \end{pmatrix} +  \sqrt{\frac{2}{3}} \cdot r_{\p} \cdot \begin{pmatrix} \cos\left(\theta_{\p} + \frac{2 \pi}{3} \right) \\ \cos(\theta_{\p}) \\ \cos\left(\theta_{\p} - \frac{2 \pi}{3} \right) \end{pmatrix}  \,,
\end{eqnarray*}
where the last equality comes from standard trigonometric identities. 
\newline

The basic partitioning of interest is into the three regions corresponding to different top two actions, namely $p_3 \geq p_{1}, p_{2}$ (players should play arms $1$ and $2$), $p_{1} \geq p_{2}, p_{3}$, and $p_{2} \geq p_{1}, p_{3}$. In cylindrical coordinates these regions are described respectively by $\theta \in \left[ \frac{\pi}{3}, \pi \right]$, $\theta \in \left[ \pi, \frac{5\pi}{3} \right]$, and $\theta \in \left[ \frac{5\pi}{3}, 2\pi \right] \cup \left[ 0, \frac{\pi}{3} \right]$.

\subsubsection{Topological difficulty} \label{sec:topological}
Intuitively, the ``topological'' difficulty of the problem is that, as $\theta$ varies continuously, the players will face a decision boundary with a collision. For example, say that in the region around $\theta=0$ (namely $\theta \in \left[ \frac{5\pi}{3}, 2\pi \right] \cup \left[ 0, \frac{\pi}{3} \right]$) we play $(i_t^A, i_t^B) = (3,1)$. As $\theta$ increases we enter the region where we should stop playing action $3$ and start playing action $2$, and thus it is natural to play $(i_t^A, i_t^B) = (2,1)$ in the region $\theta \in \left[ \frac{\pi}{3}, \pi \right]$ (i.e., only Alice is trying to figure out whether she plays action $2$ or $3$, while Bob stays constant on action $1$). On the other hand, as we decrease $\theta$ and enter the region $\theta \in \left[ \pi, \frac{5\pi}{3} \right]$, we want to play $(i_t^A, i_t^B) = (3,2)$ (i.e., it is now Bob who tries to figure out whether to play action $2$ or $1$). The problem with this construction is that at $\theta =\pi$ we go from configuration $(2,1)$ to configuration $(3,2)$, thus at this value of $\theta$ there is a constant chance of collisions! The same occurs if $(i_t^A, i_t^B) = (3,1)$. This observation is the core of our lower bound proof in Section \ref{sec:toylower}.

To fix this issue, we propose to replace this fixed interface between $(2,1)$ and $(3,2)$ by a random cut in the region $\theta \in \left[ \frac{\pi}{3}, \pi \right]$, where we will move from $(2,1)$ to $(1,2)$ (and thus at $\theta =\pi$ we move from $(1,2)$ to $(3,2)$ and there is no risk of collision). We explain this construction next (see also Figure \ref{fig_partition}).

\subsubsection{Random interface}
Let $\Theta$ be a uniform random variable in $\left[ \frac{\pi}{3}, \pi \right]$ (this is the only randomness needed by the players). We write $\mathcal{P}=\{ [m,r,\theta] | \theta=\Theta\}$, which is a (random) half-plane containing the axis $\mathcal{D}$ (this will be our ``random cut'', to be padded appropriately to move from $(2,1)$ to $(1,2)$). More precisely, we recall that $w_t=16\sqrt{\frac{\log T}{t}}$, and define the following regions:

\begin{itemize}
\item
$A_t=\{ \p=[m,r,\theta] | \frac{\pi}{3} \leq \theta < \Theta \mbox{ and } d(\p,\mathcal{P}) \geq w_t \}$,
\item
$B'_t=\{ \p=[m,r,\theta] | \frac{\pi}{3} \leq \theta < \Theta \mbox{ and } d(\p,\mathcal{P}) < w_t \} $,
\item
$C'_t=\{ \p=[m,r,\theta] | \Theta \leq \theta < \pi \mbox{ and } d(\p,\mathcal{P}) < w_t \} \setminus \cD $,
\item
$D_t=\{ \p=[m,r,\theta] | \Theta \leq \theta < \pi \mbox{ and } d(\p,\mathcal{P}) \geq w_t \}$,
\item
$B''_t=\{ \p=[m,r,\theta] | 0 \leq \theta < \frac{\pi}{3} \mbox{ or } \frac{5\pi}{3} \leq \theta < 2\pi \}$,
\item
$C''_t=\{ \p=[m,r,\theta] | \pi \leq \theta < \frac{5\pi}{3} \} \setminus \cD$.
\end{itemize}
We finally write $B_t=B'_t \cup B''_t$ and $C_t=C'_t \cup C''_t$. Note that the large or strict inequalities and the convention $\mathcal{D} \not\subset C_t$ were chosen so that $(A_t, B_t, C_t, D_t)$ is a partition of the cube $[0,1]^3$, but these choices do not really matter.

We illustrate on Figure \ref{fig_partition} the restriction of this partition to the plane of equation $p_1+p_2+p_3=\frac{3}{2}$. Note that the definition of $A_t, B_t, C_t, D_t$ does not depend on the coordinate $m$. This implies that the full partition is just obtained from Figure \ref{fig_partition} by adding one dimension orthogonally to the plane. More precisely, a point of $[0,1]^3$ belongs to a region of the partition if and only if its orthogonal projection on the plane of Figure \ref{fig_partition} belongs to that region. Note that $B''_t$ corresponds exactly to the region where the best two arms are $1$ and $3$, and $C''_t$ to the region where the best two arms are $2$ and $3$.

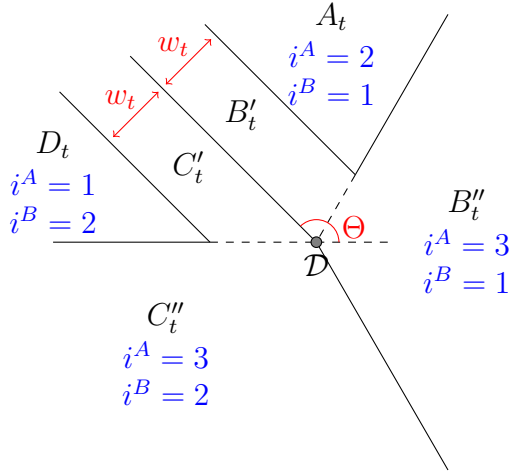
\begin{figure}
\begin{center}
\begin{tikzpicture}

\draw(-1.414,0)--(-3.5,0);
\draw(0,0)--(135:3.5);
\draw(0.518,0.897)--(60:3.5);
\draw(0,0)--(300:3.5);
\draw[dashed] (-1.414,0)--(0,0);
\draw[dashed] (0.518,0.897)--(0,0);
\draw(-1.414,0)--(-3.414,2);
\draw(0.518,0.897)--(-1.482,2.897);
\draw[dashed] (0,0)--(1,0);

\draw(0,0)node{};
\draw(0,-0.3)node[texte]{$\mathcal{D}$};

\draw(-3.5,1.3)node[texte]{$D_t$};
\draw[blue](-3.5,0.8)node[texte]{$i^A=1$};
\draw[blue](-3.5,0.3)node[texte]{$i^B=2$};
\draw(-1.7,1)node[texte]{$C'_t$};
\draw(-2,-1)node[texte]{$C''_t$};
\draw[blue](-2,-1.5)node[texte]{$i^A=3$};
\draw[blue](-2,-2)node[texte]{$i^B=2$};
\draw(-1,1.7)node[texte]{$B'_t$};
\draw(2,0.5)node[texte]{$B''_t$};
\draw[blue](2,0)node[texte]{$i^A=3$};
\draw[blue](2,-0.5)node[texte]{$i^B=1$};
\draw(0.2,3)node[texte]{$A_t$};
\draw[blue](0.2,2.5)node[texte]{$i^A=2$};
\draw[blue](0.2,2)node[texte]{$i^B=1$};

\draw[red, <->] (-2.1,2)--(-2.7,1.4);
\draw[red, <->] (-2,2.1)--(-1.4,2.7);
\draw[red] (-2.6,1.9) node[texte]{$w_t$};
\draw[red] (-1.9,2.6) node[texte]{$w_t$};
\draw[red](0.3,0) arc(0:135:0.3);
\draw[red] (0.5,0.2) node[texte]{$\Theta$};

\end{tikzpicture}
\end{center}
\caption{The restriction of our partition of the cube to the plane $\{ m_{\p}=\frac{1}{2} \}$. We recall that $B_t=B'_t \cup B''_t$ and $C_t=C'_t \cup C''_t$. The full partition is obtained from here by extending each region orthogonally to that plane. In blue, the arms played by each player in each region.}\label{fig_partition}
\end{figure}

\subsubsection{Coloring the partition}
We now define the map $\sigma_t : [0,1]^3 \rightarrow \{1,2,3\} \times \{1,2,3\}$ that the players use to select an action. It will be constant over the regions $A_t, B_t, C_t, D_t$. Precisely, as on Figure \ref{fig_partition}:
\[
\sigma_t(\q) := \begin{cases}
(2,1) & \mbox{if $\q \in A_t$,} \\
(3,1) & \mbox{if $\q \in B_t$,}\\
(3,2) & \mbox{if $\q \in C_t$,} \\
(1,2) & \mbox{if $\q \in D_t$.} \\
\end{cases}
\]
We denote by $\sigma_t^A$ and $\sigma_t^B$ the two coordinates of $\sigma_t$. For example, for $\q \in A_t$, we have $\sigma_t^A(\q)=2$ and $\sigma_t^B(\q)=1$. As explained above, the strategy is to set $i_t^A=\sigma_t^A(\q_t^A)$ and $i_t^B=\sigma_t^B(\q_t^B)$.

Roughly speaking, the reasons why this strategy works are as follows:
\begin{itemize}
\item[$\bullet$]
By \eqref{eq:distance} $q_t^A$ and $q_t^B$ are never too far away from each other, so they are either in the same region or in two neighbour regions of the partition, and the strategy ensures that there is no collision.
\item[$\bullet$]
Under the event $\Omega$ of \eqref{eq:Omega}, the players almost play the best two arms except in the region $B'_t \cup C'_t$. If $\p$ is close to the axis $\mathcal{D}$, this is not suboptimal by a lot. If $\p$ is far away from $\mathcal{D}$, then $\P \left( \p \in B'_t \cup C'_t \right)$ is small since $\Theta$ is randomized.
\end{itemize}

\subsection{Regret analysis}\label{sec:regrettoyupper}
We give here the proof of Theorem \ref{thm:toyupper}, with \eqref{eq:toyregret} holding in expectation over $\Theta$ (which is the only source of randomness in the players' strategy).

\subsubsection{No collision property}
First observe that the coloring $\sigma_t$ is such that there are no collisions for neighboring regions, i.e., if $U, V \in \{A_t, B_t, C_t, D_t\}$ are neighboring regions then $\sigma_{t}^A(U) \neq \sigma_{t}^B(V)$ and $\sigma_{t}^B(U) \neq \sigma_{t}^A(V)$. Next we note that two non-neighboring regions are well-separated.
\begin{lemma}\label{lem_separation}
In the partition $(A_t, B_t, C_t, D_t)$, the distance between any two non-neighboring regions is at least $w_t$.
\end{lemma}

\begin{proof}
The pairs of non-neighboring regions are $(A_t, D_t)$, $(A_t, C_t)$ and $(B_t, D_t)$. Any of these pairs has its two elements on different sides of the set $\left\{ \theta=\Theta \text{ or } \theta = \frac{5 \pi}{3} \right\}$. Moreover, simple geometric considerations show that $A_t$ and $D_t$ are both at distance $w_t$ from that set. Thus all these distances are at least $w_t$.
\end{proof}

Finally recall that on $\Omega$ the observations of Alice and Bob are close to each other (see \eqref{eq:distance}), so we can conclude that Alice and Bob never collide when $\Omega$ holds true.

\subsubsection{Controlling the regret from suboptimal decisions}
We denote by $B(x,r)$ the ball of radius $r$ around $x$ for the Euclidean distance. Given that there are no collisions on $\Omega$, we have:
\begin{align}\label{regret_sum_max}
R_T = \sum_{t=1}^T  (p_{i_t^A} + p_{i_t^B} - \p^* ) & = \sum_{t=1}^T  (p_{\sigma_{t}^A(q_t^A)} + p_{\sigma_{t}^B(q_t^B)} - \p^* ) \nonumber \\
 & \leq \sum_{t=1}^T  \max_{\q, \q' \in B(\p, w_t / 2)} (p_{\sigma_{t}^A(\q)} + p_{\sigma_{t}^B(\q')} - \p^* ) \nonumber \\
 & \leq 2 \sum_{t=1}^T \max_{\q \in B(\p, w_t / 2)} (p_{\sigma_{t}^A(\q)} + p_{\sigma_{t}^B(\q)} - \p^* ),
 \end{align}
where the second line uses that under $\Omega$ we have $q_t^A, q_t^B \in B(\p, w_t /2)$, and the last line uses the bound
\[ p_{\sigma_{t}^A(\q)} + p_{\sigma_{t}^B(\q')} - \p^* \leq \left( p_{\sigma_{t}^A(\q)} + p_{\sigma_{t}^B(\q)} -\p^* \right) + \left( p_{\sigma_{t}^A(\q')} + p_{\sigma_{t}^B(\q')} -\p^* \right).\]
 
To control the last quantity of \eqref{regret_sum_max}, let us first assume that $d(\p, \cP) > 2 w_t$. Then we know that for any $\q \in B(\p, w_t / 2)$, one has $\q \not\in B_t' \cup C_t'$. By construction, $q_{\sigma_{t}^A(\q)} + q_{\sigma_{t}^B(\q)} = \q^*$ for any $\q \not\in B_t' \cup C_t'$. Moreover the map $\q \mapsto \q^*$ is $2$-Lipschitz so we get that $p_{\sigma_{t}^A(\q)} + p_{\sigma_{t}^B(\q)} \leq w_t + q_{\sigma_{t}^A(\q)} + q_{\sigma_{t}^B(\q)} = w_t+\q^* \leq 2 w_t + \p^*$. In other words, so far we have proved that on $\Omega$ we have:
\[
R_T \leq 4 \sum_{t=1}^T w_t + 2 \sum_{t=1}^T  \mathbbm{1}_{d(\p, \cP) \leq 2 w_t} \max_{\q \in B(\p, w_t / 2)} (p_{\sigma_{t}^A(\q)} + p_{\sigma_{t}^B(\q)} - \p^* ) \,.
\]
Note that 
\[
p_{\sigma_{t}^A(\q)} + p_{\sigma_{t}^B(\q)} - \p^* \leq \max_{i \neq j} |p_i - p_j| \leq r_{\p} \,.
\]
Thus we get with the two above displays:
\begin{equation}\label{eqn_regret_with_proba_interface}
\E_{\Theta} R_T \leq 4 \sum_{t=1}^T w_t + 2 \sum_{t=1}^T r_{\p} \P_{\Theta}(d(\p, \cP) \leq 2 w_t) \,.
\end{equation}
The proof is now concluded with the following lemma, which implies $\E_{\Theta} R_T \leq 10 \sum_{t=1}^T w_t \leq 320 \sqrt{T \log T}$.
\begin{lemma}\label{lem_proba_interface}
For every $t$ and $\p$, we have
\begin{equation}\label{eqn_proba_interface}
\P \left( d(\p, \mathcal{P}) \leq 2 w_t \right) \leq 3 \frac{w_t}{r_{\p}}.
\end{equation}
\end{lemma}

\begin{proof}
We first note that, since the half-plane $\cP$ is orthogonal to the plane $\{ m_{\p}=\frac{1}{2} \}$ of Figure \ref{fig_partition}, both sides of \eqref{eqn_proba_interface} are unchanged if we replace $\p$ by its projection on $\{ m_{\p}=\frac{1}{2} \}$, so we can assume $\p \in \mathcal{P}$. Moreover, the distance between $\p$ and $\mathcal{P}$ is equal to the distance in $\{ m_{\p}=\frac{1}{2} \}$ between $\p$ and the half-line $\mathcal{P} \cap \{ m_{\p}=\frac{1}{2} \}$.

We also note the result is obviously true if $r_{\p} > 2 w_t$ (the right-hand side of \eqref{eqn_proba_interface} is larger than $1$), so we can assume  $r_{\p} \leq 2 w_t$. Then we have
\[ d(\p, \mathcal{P})=r_{\p} \sin \alpha \,,\]
where $\alpha$ is the angle between the line from the point $\left( \frac12, \frac12, \frac12 \right)$ to $\p$ and the half-line $\{ \theta=\Theta \}$, in the plane of Figure \ref{fig_partition}. We have $\alpha=| \Theta-\theta_{\p}|$, so the event of \eqref{eqn_proba_interface} is equivalent to
\[ \theta_{\p}-\arcsin \frac{2 w_t}{r_{\p}} \leq \Theta \leq \theta_{\p}+ \arcsin \frac{2 w_t}{r_{\p}} \,.\]
This has probability $\frac{3}{2\pi} \times 2 \arcsin \frac{2 w_t}{r_{\p}} \leq 3 \frac{w_t}{r_{\p}}$, which concludes the proof of the lemma.
\end{proof}

\subsection{Derandomization via a dynamic interface} \label{sec:dynamic}
The only place where we used the randomness in $\Theta$ is Lemma \ref{lem_proba_interface}. To derandomize the algorithm, we can replace the random angle $\Theta$ by a deterministic, time-dependent angle $\left( \theta_t \right)_{t \in [T]}$, with $\frac{\pi}{3} \leq \theta_t \leq \pi$. In this setting, all the proof is the same until \eqref{eqn_regret_with_proba_interface}, which becomes
\[ R_T \leq 4 \sum_{t=1}^T w_t + 2 r_{\p} \sum_{t=1}^T \mathbbm{1}_{d(\p, \cP_t) \leq 2 w_t},\]
where $\cP_t=\{ [m,r,\theta] | \theta=\theta_t\}$. For the same reason as in the proof of Lemma \ref{lem_proba_interface}, if $d(\p, \cP_t) \leq 2 w_t$, then $|\theta_{\p}-\theta_t| \leq \arcsin \frac{2w_t}{r_{\p}} \leq \pi \frac{w_t}{r_{\p}}$. Therefore, to obtain the analog of Lemma \ref{lem_proba_interface}, it is enough to find $(\theta_t)$ such that, for any $r$ and $\theta$, the number of $t$ such that $|\theta-\theta_t| \leq \pi \frac{w_t}{r}$ is at most $\frac{3}{r} \sum_{t=1}^T w_t$.

One way to do so is the following: for every $t$, let $k$ be such that $2^k \leq t < 2^{k+1}$, and take
\[\theta_t=\frac{\pi}{3}+\frac{2\pi}{3} \frac{t-2^k}{2^k}.\]
In that case, for every fixed $k$, $r$ and $\theta$, using that $w_t$ is decreasing in $t$, we have
\[ \sum_{t=2^k}^{2^{k+1}-1} \mathbbm{1}_{|\theta_t-\theta| \leq \frac{\pi w_t}{r}} \leq \sum_{t=2^k}^{2^{k+1}-1} \mathbbm{1}_{|\theta_t-\theta| \leq \frac{\pi w_{2^k}}{r}} \leq 1+\frac{\pi w_{2^k}/r}{2\pi/(3 \times 2^k)}=1+\frac{3}{2} \times 2^k \frac{w_{2^k}}{r} \leq 3 \sum_{t=2^k}^{2^{k+1}-1} \frac{w_t}{r},\]
and summing over $k$ yields the result.

\section{Bandit upper bound} \label{sec:banditupper}
We prove here Theorem \ref{thm:main}. The extra difficulty introduced by the bandit setting compared to the full-information toy model is that, in addition to coordinating for exploitation (which is the key point of the toy model), the players also have to coordinate their {\em exploration} of the arms. Moreover, there needs to be a {\em smooth} transition between exploration and exploitation, so that there are also no collisions if one player stops exploring before the other. To do so we introduce extra padding around the decision boundaries of the partition built in the previous section, and we give a carefully choreographed dynamic coloring of this new partition. An explicit algorithm is fully described below by combining the definition~\eqref{eqn_defn_qt}, the partition constructed in Section~\ref{subsec_partition} (and represented on Figure~\ref{fig_partition_bandits}) and the table on Figure~\ref{table_strategy_bandits}.
\newline

We denote $w_t = 2^{15} \sqrt{\frac{\log(T)}{t}}$. For $1 \leq t \leq T$, $i \in \{1,2,3\}$ and $X \in \{A,B\}$, we denote by $n_t^X(i)$ the number of times from $1$ to $t-1$ where player $X$ has played arm $i$. We also write
\begin{equation}\label{eqn_defn_qt}
q_t^X(i)=\frac{1}{n_t^X(i)} \sum_{\substack{i=1 \\ i_t^X=i}}^{t-1} \max\left(\ell_t(i), \mathbbm{1}_{i_t^A = i_t^B}\right) \,,
\end{equation}
with the convention $q_t^X(i)=0$ if $n_t^X(i)=0$. Then $\q_t^X=\left( q_t^X(1), q_t^X(2), q_t^X(3)\right)$ is an estimate at time $t$, according to player $X$, of $\p$. Note that this estimator is biased due to the potential collisions. This issue will be handled below (Lemma \ref{lem:nocollisionsbandit}).

We will prove the absence of collisions by induction on $t$, which means that we need to show that our estimators at time $t$ are not too bad if there has been no collision before. For this reason, we define the following event:
\begin{multline*}
\Omega = \Big\{\forall t \in [T], i \in \{1,2,3\}, X\in\{A,B\}, \mbox{ if there has been no collision} \\ \mbox{at times $1, \dots, t-1$, then } \left| q_t^X(i) - p_i \right| < \frac{w_{4 n_t^X(i) + 5}}{32} \Big\} \,.
\end{multline*}
If there has been no collision before time $t$, we have $q_t^X(i)=\frac{1}{n_t^X(i)} \sum_{i=1, i_t^X=i}^{t-1} \ell_t(i)$. Note that $\Omega$ depends on the $n_t^X(i)$, and therefore on the strategies used by the players. However, for any strategy, if we fix an arm $i$ and list the values $\ell_t(i)$ observed by a player $X \in \{A,B\}$, then these values are i.i.d. Bernoulli with parameter $p_i$. Therefore, the Hoeffding inequality and a union bound show that $\P \left( \Omega \right) \geq 1-\frac{1}{T}$ for any deterministic strategy of $A$ and $B$, and thereforee also for a random one. We will later prove the following result, which implies that the assumption of no collisions in $\Omega$ can be removed.

\begin{lemma} \label{lem:nocollisionsbandit}
On the event $\Omega$, our proposed bandit strategy satisfies $i_t^A \neq i_t^B$ for all $t \in [T]$.
\end{lemma}

Like in the full-information toy model, in the remainder of this section we fix loss sequences such that $\Omega$ holds true, and all probabilities are with respect to the random interface defined by $\Theta$ (see below).

\subsection{The bandit partition}\label{subsec_partition}

We recall that $w_t = 2^{15} \sqrt{\frac{\log(T)}{t}}$.
We denote by $\PP$ the half-plane $\{ \theta=\Theta \}$ and by $\QQ_1$ (resp. $\QQ_2$, $\QQ_3$) the half-plane $\{ \theta=\frac{\pi}{3} \}$ (resp. $\{ \theta=\pi \}$, $\{ \theta=\frac{5\pi}{3} \}$). We now define the following sets, that we will refer to as \emph{regions}:
\begin{itemize}
\item[$\bullet$]
$E_t=\left\{ \p | \frac{\pi}{3} \leq \theta_{\p} < \Theta \mbox{ and } d(\p, \QQ_1) \geq \frac{w_t}{2} \mbox{ and } d(\p, \PP) \geq \frac{3 w_t}{2} \right\}$,
\item[$\bullet$]
$G_t=\left\{ \p | \theta_{\p} \in \left[ 0, \frac{\pi}{3} \right) \cup \left[ \frac{5\pi}{3}, 2\pi \right) \mbox{ and } d(\p, \QQ_1) \geq \frac{w_t}{2} \mbox{ and } d(\p, \QQ_3) \geq \frac{w_t}{2} \right\}$,
\item[$\bullet$]
$H_t=\left\{ \p | d(\p, \PP \cup \QQ_3) < \frac{w_t}{2} \right\}$,
\item[$\bullet$]
$I_t=\left\{ \p | \theta_{\p} \in \left[ \pi, \frac{5\pi}{3} \right) \mbox{ and } d(\p, \QQ_2) \geq \frac{w_t}{2} \mbox{ and } d(\p, \QQ_3) \geq \frac{w_t}{2} \right\}$,
\item[$\bullet$]
$K_t=\left\{ \p | \Theta \leq \theta_{\p} < \pi \mbox{ and } d(\p, \QQ_2) \geq \frac{w_t}{2} \mbox{ and } d(\p, \PP) \geq \frac{3 w_t}{2} \right\}$,
\item[$\bullet$]
$F_t=\left\{ \p | \theta_{\p} \in \left[ 0, \Theta \right) \cup \left[ \frac{5\pi}{3}, 2\pi \right) \right\} \backslash (E_t \cup G_t \cup H_t)$,
\item[$\bullet$]
$J_t=\left\{ \p | \theta_{\p} \in \left[ \Theta, \frac{5\pi}{3} \right) \right\} \backslash (H_t \cup I_t \cup K_t)$.
\end{itemize}

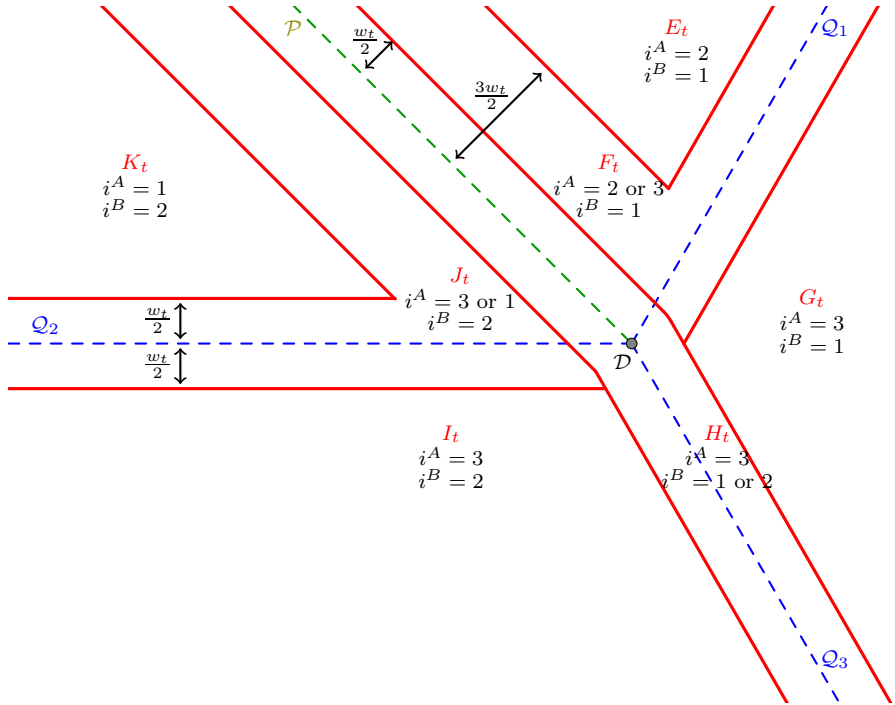
\begin{figure}
\begin{center}
\begin{tikzpicture}[line cap=round,line join=round,x=1.0cm,y=1.0cm, scale=0.6]
\clip(-13.8,-7.97) rectangle (13.2,7.47);
\draw [line width=0.8pt,dash pattern=on 4pt off 4pt,color=qqqqff,domain=0.0:13.203361792227296] plot(\x,{(-0--4.33*\x)/2.5});
\draw [line width=0.8pt,dash pattern=on 4pt off 4pt,color=qqqqff,domain=-13.804902670582594:0.0] plot(\x,{(-0-0*\x)/-5});
\draw [line width=0.8pt, dash pattern=on 4pt off 4pt,color=qqqqff,domain=0.0:13.203361792227296] plot(\x,{(-0-4.33*\x)/2.5});
\draw [line width=1.2pt,color=ffqqqq,domain=-13.804902670582594:-5.242640687119287] plot(\x,{(--8.17--1.92*\x)/-1.92});
\draw [line width=1.2pt,color=ffqqqq,domain=-13.804902670582594:-5.242640687119287] plot(\x,{(-2.51-0*\x)/-2.51});
\draw [line width=1.2pt,color=ffqqqq,domain=-13.804902670582594:-0.8001991549907305] plot(\x,{(--3.65--2.58*\x)/-2.58});
\draw [line width=1.2pt,color=ffqqqq,domain=-0.8001991549907305:13.203361792227296] plot(\x,{(-3.08-2.67*\x)/1.54});
\draw [line width=1.2pt,color=ffqqqq,domain=-13.804902670582594:-0.5773502691896253] plot(\x,{(--2.67-0*\x)/-2.67});
\draw [line width=1.2pt,color=ffqqqq,domain=-13.804902670582594:0.8208634630462484] plot(\x,{(-6.04--1.42*\x)/-1.42});
\draw [line width=1.2pt,color=ffqqqq,domain=0.8208634630462484:13.203361792227296] plot(\x,{(--1.65--1.43*\x)/0.82});
\draw [shift={(0,0)},line width=1.2pt,color=ffqqqq]  plot[domain=0.52:0.79,variable=\t]({1*1*cos(\t r)+0*1*sin(\t r)},{0*1*cos(\t r)+1*1*sin(\t r)});
\draw [line width=1.2pt,color=ffqqqq,domain=-13.804902670582594:0.7071067811865501] plot(\x,{(-2.52--1.78*\x)/-1.78});
\draw [line width=1.2pt,color=ffqqqq,domain=0.8660254037844376:13.203361792227296] plot(\x,{(--3.29-2.85*\x)/1.64});
\draw [line width=1.2pt,color=ffqqqq,domain=1.1547005383792501:13.203361792227296] plot(\x,{(-1.77--1.53*\x)/0.88});
\draw [line width=0.8pt, dash pattern=on 4pt off 4pt,color=qqzzqq,domain=-13.804902670582594:0.0] plot(\x,{(-0--2.65*\x)/-2.65});
\begin{scriptsize}
\fill [color=uuuuuu] (0,0) circle (1.5pt);
\fill [color=uuuuuu] (0,0) circle (1.5pt);
\fill [color=uuuuuu] (0,0) circle (1.5pt);
\fill [color=uuuuuu] (0,0) circle (1.5pt);
\fill [color=uuuuuu] (0,0) circle (1.5pt);
\draw[<->, thick] (-10,0.1)--(-10,0.9);
\draw[<->, thick] (-10,-0.1)--(-10,-0.9);
\draw[<->, thick] (-5.9,6.1)--(-5.3,6.7);
\draw[<->, thick] (-3.9,4.1)--(-2,6);

\draw(-10.5,0.5)node[texte]{$\frac{w_t}{2}$};
\draw(-10.5,-0.5)node[texte]{$\frac{w_t}{2}$};
\draw(-5.9,6.7)node[texte]{$\frac{w_t}{2}$};
\draw(-3.1,5.5)node[texte]{$\frac{3 w_t}{2}$};

\draw(1,7)[red]node[texte]{$E_t$};
\draw(1,6.5)node[texte]{$i^A=2$};
\draw(1,6)node[texte]{$i^B=1$};
\draw(-0.5,4)[red]node[texte]{$F_t$};
\draw(-0.5,3.5)node[texte]{$i^A=2 \mbox{ or } 3$};
\draw(-0.5,3)node[texte]{$i^B=1$};
\draw(4,1)[red]node[texte]{$G_t$};
\draw(4,0.5)node[texte]{$i^A=3$};
\draw(4,0)node[texte]{$i^B=1$};
\draw(1.9,-2)[red]node[texte]{$H_t$};
\draw(1.9,-2.5)node[texte]{$i^A=3$};
\draw(1.9,-3)node[texte]{$i^B=1 \mbox{ or } 2$};
\draw(-4,-2)[red]node[texte]{$I_t$};
\draw(-4,-2.5)node[texte]{$i^A=3$};
\draw(-4,-3)node[texte]{$i^B=2$};
\draw(-3.8,1.5)[red]node[texte]{$J_t$};
\draw(-3.8,1)node[texte]{$i^A=3\mbox{ or } 1$};
\draw(-3.8,0.5)node[texte]{$i^B=2$};
\draw(-11,4)[red]node[texte]{$K_t$};
\draw(-11,3.5)node[texte]{$i^A=1$};
\draw(-11,3)node[texte]{$i^B=2$};

\draw(0,0)node{};
\draw(-0.2,-0.4)node[texte]{$\mathcal{D}$};
\draw[blue](4.5,7)node[texte]{$\QQ_1$};
\draw[blue](-13,0.4)node[texte]{$\QQ_2$};
\draw[blue](4.5,-7)node[texte]{$\QQ_3$};
\draw[olive](-7.5,7)node[texte]{$\PP$};
\end{scriptsize}
\end{tikzpicture}
\end{center}
\caption{The intersection of our partition with the plane $\{p_1+p_2+p_3=\frac{3}{2}\}$. Below the names of the regions are the arms played by the players in the first two columns of the table, i.e. for $t \equiv 1$ or $2$ modulo $4$ (for $t \equiv 3$ or $0$ modulo $4$, the roles of players $A$ and $B$ are exchanged).}\label{fig_partition_bandits}
\end{figure}

As in the full information case, we have represented on Figure \ref{fig_partition_bandits} the restriction of this partition to the plane $\left\{ p_1+p_2+p_3=\frac{3}{2} \right\}$. Here again, since that plane is orthogonal to the half-planes $\PP$, $\QQ_1$, $\QQ_2$, $\QQ_3$, the full partition is obtained by extending Figure \ref{fig_partition_bandits} orthogonally to its plane.

\subsection{Dynamic coloring}
The strategy is now the following: for every $0 \leq t < \frac{T}{4}$, player $A$ will decide according to the region of $\q^A_{4t+1}$ where he plays at times $4t+1$, $4t+2$, $4t+3$, $4t+4$, and similarly for player $B$ according to the region of $\q_{4t+1}^B$. More precisely, player $A$ will play according to the table below. The way to read this table is as follows: if $2/1$ is written at the intersection of the row "$\q_{4t+1}^A / \q_{4t+1}^B \in E_{4t+1}$" and the column "$4t+2$", this means that if $\q_{4t+1}^A \in E_{4t+1}$, then player $A$ plays arm $2$ at time $4t+2$. If $\q_{4t+1}^B \in E_{4t+1}$, then player $B$ plays arm $1$ at time $4t+2$.

\begin{figure}[h!]
\begin{center}
\begin{tabular}{|c | c | c | c | c |} 
\hline
 & $4t+1$ & $4t+2$ & $4t+3$ & $4t+4$ \\
\hline
$\q_{4t+1}^A / \q_{4t+1}^B \in E_{4t+1}$ & 2 / 1 & 2 / 1 & 1 / 2 & 1 / 2 \\ 
\hline
$\q_{4t+1}^A / \q_{4t+1}^B \in F_{4t+1}$ & 2 / 1 & 3 / 1 & 1 / 2 & 1 / 3 \\
\hline
$\q_{4t+1}^A / \q_{4t+1}^B \in G_{4t+1}$ & 3 / 1 & 3 / 1 & 1 / 3 & 1 / 3 \\
\hline
$\q_{4t+1}^A / \q_{4t+1}^B \in H_{4t+1}$ & 3 / 1 & 3 / 2 & 1 / 3 & 2 / 3 \\
\hline
$\q_{4t+1}^A / \q_{4t+1}^B \in I_{4t+1}$ & 3 / 2 & 3 / 2 & 2 / 3 & 2 / 3 \\
\hline
$\q_{4t+1}^A / \q_{4t+1}^B \in J_{4t+1}$ & 3 / 2 & 1 / 2 & 2 / 3 & 2 / 1 \\
\hline
$\q_{4t+1}^A / \q_{4t+1}^B \in K_{4t+1}$ & 1 / 2 & 1 / 2 & 2 / 1 & 2 / 1 \\
\hline
\end{tabular}
\end{center}
\caption{The table describing the arms played by the players at time $4t+1, \dots, 4t+4$ according to $\q_{4t+1}^A$ and $\q_{4t+1}^B$.}
\label{table_strategy_bandits}
\end{figure}

Although it might seem quite complicated, this table is actually a natural adaptation of the full information strategy, where we have "smoothened" the boundaries between regions. Let us first focus on the first two columns: the regions $E$, $G$, $I$ and $K$ then correspond to the regions $A$, $B$, $C$, $D$ of the full information strategy. The difference here is that, if for example we are in the region where $p_2$ and $p_3$ are close but much larger than $p_1$, it is necessary to explore both arms $2$ and $3$ during a long time to find which is the best one. This is the role of region $F$, and regions $H$ and $J$ play a similar role.

Moreover, the last two columns are the same as the first two, where the roles of $A$ and $B$ have been exchanged. This is necessary to make sure that each of the players has information about all the arms. Of course, such a problem did not exist in the full information case.

\begin{remark}
It might have seemed more natural to choose the arm played at time $4t+2$ according to $\q^A_{4t+2}$ instead of $\q^A_{4t+1}$. The reason why we chose not to do so is to make sure that, as long as $\q^A$ belongs to $F_t \cup H_t \cup J_t$, player $A$ plays all the arms regularly, even if $\q^A$ "oscillates" for example between $F_t$ and $H_t$.
\end{remark}

\subsection{Exploration phase and no collision property}
The regions $F_t$, $H_t$ and $J_t$ can be considered as "exploration" regions, since they are regions where both players play the three arms. It is immediate from the definition of the regions that $E_t$, $G_t$, $I_t$ and $K_t$ are increasing in $t$, which means that $F_t \cup H_t \cup J_t$ is decreasing in $t$. Therefore, it is natural to expect that $\q_t^A$ will be in $F_t \cup H_t \cup J_t$ in the beginning ("exploration phase"), and in the complementary after some time ("exploitation phase"). We make this intuition precise in the proof of the next lemma.

\begin{lemma}\label{lem_close_or_same_region}
Under $\Omega$, for every $1 \leq t \leq T$, if there has been no collision before time $t$, then either $\p$, $\q_t^A$ and $\q_t^B$ are in the same region, or $\q_t^A$, $\q_t^B$ belong to the ball of radius $\frac{w_t}{4}$ around $\p$.
\end{lemma}

\begin{proof}
For $X \in \{A,B\}$, we denote by $\tau^X$ the first time $t$ such that $q_t^X \notin F_t \cup H_t \cup J_t$, with the convention $\tau^X=+\infty$ if $q_t^X \in F_t \cup H_t \cup J_t$ for all $t \in [T]$. In particular, for any $s < \frac{\tau^X-5}{4}$ we have
$\q^X_{4s+1} \in F_t \cup H_t \cup J_t$, which means that each arm appears at least once among $i^X_{4s+1}, i^X_{4s+2}, i^X_{4s+3}, i^X_{4s+4}$. Therefore, we must have $n_t^X(i) \geq \frac{\min(t, \tau^X)-5}{4}$ for every arm $i$. Using the event $\Omega$, this implies $\left|q_t^X(i) - p_i \right| < \frac{w_{\min(t, \tau^X)}}{32}$ for all $i$, and thus
\begin{equation} \label{eq:banditdist}
d(\q_t^X, \p) < \frac{w_{\min(t, \tau^X)}}{16} \,.
\end{equation}
In particular, since any point at distance $\leq \frac{w_t}{2}$ from $\QQ_1 \cup \QQ_2 \cup \QQ_3 \cup \PP$ is in $F_t \cup H_t \cup J_t$, we have $d(\q_{\tau^X}^X, \QQ_1 \cup \QQ_2 \cup \QQ_3 \cup \PP) \geq \frac{w_t}{2}$. Hence $\p$ must be at distance at least $\frac{7}{16} w_{\tau^X}$ from $\QQ_1 \cup \QQ_2 \cup \QQ_3 \cup \PP$ (it is also immediate if $\tau_X=+\infty$). Next observe that for $t \geq 16 \tau^X$ one has $\frac{7}{16} w_{\tau^X} \geq \frac{3}{2} w_{t} + \frac1{16} w_{\tau^X}$. Since $F_t \cup H_t \cup J_t$ lie entirely at distance at most $\frac{3}{2} w_t$ from $\QQ_1 \cup \QQ_2 \cup \QQ_3 \cup \PP$, we deduce that $\p$ is at distance $\frac{w_{\tau^X}}{16}$ from $F_t \cup H_t \cup J_t$, so the ball of center $\p$ and radius $\frac{w_{\tau^X}}{16}$ is contained in the region of $\p$ (which may be $E_t$, $G_t$, $I_t$ or $K_t$). By \eqref{eq:banditdist}, this implies that $\q_t^X$ is in the same region as $\p$.

On the other hand, for $t \leq 16 \tau^X$, \eqref{eq:banditdist} gives
\[
d(\q_t^X, \p) < \frac{w_{t/16}}{16} = \frac{1}{4} w_t \, ,
\] 
which concludes the proof.
\end{proof}

We now prove the no collision property. Note that this will allow us to use the event $\Omega$ without having to assume that there has been no collision so far.

\begin{proof}[Proof of Lemma \ref{lem:nocollisionsbandit}.]
As explained earlier, we assume $\Omega$ and prove by induction on $t$ the absence of collisions until $t$. Assume there was no collision at times $1, \dots, t-1$. By Lemma \ref{lem_close_or_same_region}, we know that for every $t$, either $\q_t^A$ and $\q_t^B$ lie in the same region, or $d(\q_t^A, \q_t^B) < \frac{w_t}{2}$. In the first case, there is no collision.

In the second case, we call two regions \emph{compatible} if, whenever $A$ plays according to the first one and $B$ according to the second, we have $i_t^A \ne i_t^B$. By looking at the table of Figure \ref{table_strategy_bandits}, we find that compatible regions are given by the following graph, where two regions are linked by an edge if they are compatible.
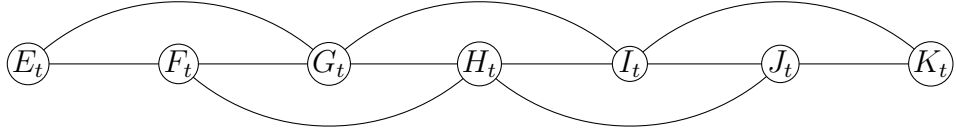
\begin{figure}[!h]
\begin{center}
\begin{tikzpicture}
\draw(0,0)--(12,0);
\draw(0,0)to[bend left=45](4,0);
\draw(4,0)to[bend left=45](8,0);
\draw(8,0)to[bend left=45](12,0);
\draw(2,0)to[bend right=45](6,0);
\draw(6,0)to[bend right=45](10,0);

\draw(0,0)node[white]{$E_t$};
\draw(2,0)node[white]{$F_t$};
\draw(4,0)node[white]{$G_t$};
\draw(6,0)node[white]{$H_t$};
\draw(8,0)node[white]{$I_t$};
\draw(10,0)node[white]{$J_t$};
\draw(12,0)node[white]{$K_t$};
\end{tikzpicture}
\end{center}
\end{figure}

By the definitions of the regions, the distance between any two non-compatible regions is always at least $w_t$ (this is very similar to Lemma \ref{lem_separation} in the full information case, so we omit the detailed proof). Therefore, no collision can happen if the event of Lemma \ref{lem_close_or_same_region} occurs, which proves the lemma.
\end{proof}

\subsection{Concluding the proof of Theorem \ref{thm:main}}
For every $t$, we write $\ut=4 \lfloor \frac{t-1}{4} \rfloor +1$, so that $i_t^A$ is chosen according to the region of $\q_{\underline{t}}^A$. We denote by $\sigma_t=\left( \sigma_t^A, \sigma_t^B \right)$ the map prescribed by the table of Figure \ref{table_strategy_bandits}, so that $i_t^X=\sigma_t^X(\q_{\ut}^X)$. Using the fact that we have no collisions, we have
\[
R_T = \sum_{t=1}^T \left( p_{\sigma_t^A(\q_{\ut}^A)} + p_{\sigma_t^B(\q_{\ut}^B)} - \p^* \right) \,.
\]
Just like in the full information case (Section \ref{sec:regrettoyupper}) we decompose the sum into two terms, based on whether $d(\p, \cP) > 2 w_{\ut}$ or not. The case where $d(\p, \cP) \leq 2 w_{\ut}$ is dealt exactly as in the full information case, and gives a term $6 \sum_{t=1}^T w_{\ut}$ in expectation over $\Theta$. Now for the other term, we assume that $d(\p, \cP) > 2 w_{\ut}$ and we write, thanks to the dichotomy given by Lemma \ref{lem_close_or_same_region},
\begin{eqnarray*}
 p_{\sigma_t^A(\q_{\ut}^A)} + p_{\sigma_t^B(\q_{\ut}^B)} - \p^* & \leq & p_{\sigma_{t}^A(\p)} + p_{\sigma_{t}^B(\p)} - \p^* + 2 \max_{\q \in B(\p, w_{\ut} /4)} \left( p_{\sigma_{t}^A(\q)} + p_{\sigma_{t}^B(\q)} - \p^* \right) \\
 & \leq & 3 \max_{\q \in B(\p, w_{\ut} /4)} \left( q_{\sigma_{t}^A(\q)} + q_{\sigma_{t}^B(\q)} - \q^* \right) + 3 w_{\ut} \,,
\end{eqnarray*}
where the second inequality uses that $\q \mapsto \q^*$ is $2$-Lipschitz. Finally it only remains to observe that the construction of the bandit partition is such that for any $\q$ with $d(\q, \cP) \geq \frac{3 w_{\ut}}{2}$ one has
\[
q_{\sigma_t^A(\q)} + q_{\sigma_t^B(\q)} - \q^* \leq w_{\ut} \,.
\]
Thus we have proved that, $R_T \mathbbm{1}_{d(\p, \cP) > 2 w_{\ut}} \leq 6 \sum_{t=1}^T w_{\ut}$, and $\E_{\Theta} [ R_T \mathbbm{1}_{d(\p, \cP)) \leq 2 w_{\ut}} ] \leq 6 \sum_{t=1}^T w_{\ut}$. The expected regret is therefore bounded by $12 \sum_{t=1}^T w_{\ut}=O(\sqrt{T \log T})$, which concludes the proof of Theorem \ref{thm:main}.

\begin{remark}\label{rk_footnote_thm}
Let us finally justify the footnote in Theorem \ref{thm:main}. The only issue in our current proof is that the event $\Omega$ depends on the strategy of the players, and therefore on $\Theta$. One way to handle this is to take $\Theta$ to be a uniform variable among the multiples of $\frac{1}{T}$ in $\left[ \frac{\pi}{3}, \pi \right]$ instead of a uniform variable on the interval $\left[ \frac{\pi}{3}, \pi \right]$. The Hoeffding inequality and a union bound over $\Theta$ then guarantee that, with probability at least $1-1/T$, the event $\Omega$ holds simultaneously for all values of $\Theta$. The rest of the proof can be easily adapted, up to irrelevant rounding issues.
\end{remark}

\section{Toy model lower bound} \label{sec:toylower}
We prove here Theorem \ref{thm:lower}. The goal is essentially to exploit the topological obstruction we alluded to in Section \ref{sec:topological}. This topological obstruction is basically Lemma \ref{lem_cyclic_argument}.

\subsection{The hard instance}
We first describe the law of $(p_1,p_2,p_3)$. Let $\eps>0$ be small (it is actually enough to have $\eps<1/4$). Let $I$ be the following union of intervals:
\begin{align*}
I=\left[ \frac{\pi}{3}-2T^{-1/2+2\eps}, \frac{\pi}{3}+2T^{-1/2+2\eps} \right] & \cup \left[ \frac{3\pi}{3}-2T^{-1/2+2\eps}, \frac{3\pi}{3}+2T^{-1/2+2\eps} \right]\\
& \cup \left[ \frac{5\pi}{3}-2T^{-1/2+2\eps}, \frac{5\pi}{3}+2T^{-1/2+2\eps} \right],
\end{align*}
with total measure $12T^{-1/2+2\eps}$. We assume $T^{-1/2+2\eps} < \frac{\pi}{6}$ so that the definition makes sense. Let $\Theta$ be a random variable on $[0,2\pi]$ with distribution
\begin{equation}\label{law_theta}
\frac{1}{4\pi} \mathrm{d}\theta + \frac{\mathbbm{1}_{\theta \in I}}{24 T^{-1/2+2\eps}} \mathrm{d}\theta.
\end{equation}
In other words $\Theta$ is picked uniformly in $[0,2\pi]$ with probability $\frac{1}{2}$ and uniformly in $I$ with probability $\frac{1}{2}$.

Finally using the cylindrical coordinates $\p = [m_{\p}, r_{\p}, \theta_{\p}]$ from Section \ref{sec:cylindrical} we set $m_{\p}=1/2$, $r_{\p} = \sqrt{\frac{3}{2}} T^{-\epsilon}$, and $\theta_{\p} = \Theta$. We also denote by $\left( p_1(\Theta), p_2(\Theta), p_3(\Theta) \right)$ the Cartesian coordinates of $\p$, and write $p^*(\Theta)$ for the sum of the two smallest coordinates.

In particular $(p_1, p_2, p_3)$ is picked on a circle. Moreover, the "reinforcement" near $\frac{\pi}{3}$, $\pi$ and $\frac{5\pi}{3}$ of the law of $\Theta$ implies that the law of $(p_1, p_2, p_3)$ is reinforced at the places where two $p_i$ are almost equal, and much larger than the third.

\subsection{Proof skeleton}
From now on, we assume that $A$ and $B$ follow a fixed, deterministic strategy. We concentrate on the quantity:
\[
r_t(\theta) = \E \left[ 2 \cdot \mathbbm{1}_{i_t^A = i_t^B} + \mathbbm{1}_{i_t^A \neq i_t^B} (p_{i_t^A} + p_{i_t^B}) - \p^* \bigg) \big| \Theta = \theta \right] \,.
\]
It is easy to see (and the standard route for bandit lower bounds) that it is sufficient to prove that, for every $1 \leq t \leq T$, we have
\begin{equation}\label{regret_at_one_step}
\E \left[ r_t(\Theta) \right] \geq c \sqrt{\frac{\log T}{T}} \,.
\end{equation}
Therefore, we fix such a $t$ until the end of the proof. Key quantities of interest will be the following functions, defined for $i \in \{1,2,3\}$ and $X \in \{A,B\}$:
\[
f_{i}^{X}(\theta)=\P \left( i_t^X=i | \Theta=\theta \right) \,.
\]
Even if this depends on $t$, since $t$ is fixed until the end of the proof, we drop the $t$ in the notation. Since the loss vectors observed by $A$ and $B$ are independent conditionally on $ \Theta$, we can write
\begin{equation}\label{decomposition_rt}
r_t(\theta)=\sum_{i=1}^3 f_{i}^{A}(\theta)f_{i}^{B}(\theta) (2-p^*(\theta)) + \sum_{i \ne j} f_{i}^{A}(\theta)f_{j}^{B}(\theta) (p_i(\theta)+p_j(\theta)-p^*(\theta)) \geq 0 \,.
\end{equation}

The proof will now proceed by analyzing properties of the functions $f_i^X$, in particular the various constraints they must satisfy for the players to hope for a small regret.

\subsection{Constraints on the functions $f_i^X$}

In our proof, the fact that the players cannot have a very precise estimate of $\Theta$ will be encoded by the fact that the functions $f_i^A, f_i^B$ are smooth enough, so that the players cannot change drastically their choices when $\theta$ varies a little. Therefore, the first step is to prove an estimate on the regularity of the functions $f_i^A$, $f_i^B$.

\begin{lemma}\label{lem_regularity_f}
The functions $f_i^A$ and $f_i^B$ are analytic. Moreover, let $\delta>0$. Then there is a constant $c>0$ (depending on $\delta$ but not on $t$ or $T$) such that, for every $\theta, \theta'$, we have
\[f_i^A(\theta') \geq \left( f_i^A(\theta)-\delta \right) \exp \left( -c-c T^{1-2\eps} |\theta'-\theta|^2 \right),\]
and the same is true for $f_i^B$.
\end{lemma}

\begin{proof}
Both functions are polynomials in $(p_1, p_2, p_3)$, so they are analytic in $\delta$.

For the second point, we start by defining a "truncation" of the functions $f_i^A$. If $E$ is an event, we write
\[f_i^A(\theta,E)=\P \left( i_t^A=i \mbox{ and } E \mbox{ occurs} |\Theta=\theta \right).\]
We fix a constant $C$, and denote by $E_C(\theta)$ the event that $\left|\sum_{s=1}^{t-1}\ell_s^A(i)-(t-1) p_i(\theta) \right| \leq C \sqrt{T}$ for every $j \in \{1,2,3\}$. By the central limit theorem, if $C$ is chosen large enough (independently of $\theta$, $t$ and $T$), we have $\P \left( E_C(\theta) \right) \geq 1-\delta$, so
\[ f_i^A \left( \theta, E_C(\theta) \right) \geq f_i^A(\theta)-\delta.\]
On the other hand, we obviously have $f_i^A(\theta') \geq f_i^A \left( \theta', E_C(\theta) \right)$, so it is enough to prove
\begin{equation}\label{truncatd_continuity}
f_i^A \left( \theta', E_C(\theta) \right) \geq f_i^A \left( \theta, E_C(\theta) \right) \exp \left( -c-c T^{1-2\eps} |\theta'-\theta|^2 \right).
\end{equation}
For this, let $\ell= \left( \ell_s(i) \right)_{1 \leq s \leq t-1, 1 \leq i \leq 3} \in \left( \{0,1\}^3 \right)^{t-1}$ be a possible value of the loss vectors observed by $A$ until time $t-1$. For $j \in \{1,2,3\}$, we write $S(j)=\sum_{s=1}^{t-1} \ell_s(j)$. Then we have
\[ \log \frac{\P \left( \mbox{$A$ observes $\ell$} |\Theta=\theta' \right)}{\P \left( \mbox{$A$ observes $\ell$} |\Theta=\theta \right)} = \sum_{j=1}^3 \left( S(j) \log \frac{p_j(\theta')}{p_j(\theta)} + \left( t-1-S(j) \right) \log \frac{1-p_j(\theta')}{1-p_j(\theta)} \right).\]
The ratio $\frac{p_j(\theta')}{p_j(\theta)}$ is going to $1$ as $T \to +\infty$, uniformly in $\theta$, so we can use the inequality $\log (1+x) \geq x-x^2$ to bound the above quantity from below by
\begin{equation}\label{log_radon_nikodym}
\sum_{j=1}^3  \left( p_j(\theta')-p_j(\theta) \right) \left( \frac{S(j)}{p_j(\theta)} - \frac{t-1-S(j)}{1-p_j(\theta)} \right) - \sum_{j=1}^3 |p_j(\theta')-p_j(\theta)|^2 \left( \frac{S(j)}{p_j(\theta)^2} + \frac{t-1-S(j)}{\left( 1-p_j(\theta) \right)^2} \right).
\end{equation}
The second term is bounded from below by
\[- \sum_{j=1}^3 \left| p_j(\theta')-p_j(\theta) \right|^2 \times \frac{2t}{1/16} \geq -96 T^{1-2\eps} |\theta'-\theta|^2, \]
by using $\frac{1}{4} \leq p_j(\theta) \leq \frac{3}{4}$, and then $\left| \frac{\mathrm{d}p_j(\theta)}{\mathrm{d}\theta} \right| \leq T^{-\eps}$ and $t \leq T$.

On the other hand, since we work on the event $E_C(\theta)$, we have $\left| S(j) - (t-1)p_j(\theta) \right| \leq C \sqrt{T}$, so both $\frac{S(j)}{p_j(\theta)}$ and $\frac{t-1-S(j)}{1-p_j(\theta)}$ are close to $t-1$. More precisely, we can bound the absolute value of the first sum of \eqref{log_radon_nikodym} by
\[ \sum_{j=1}^3 \left| p_j(\theta')-p_j(\theta) \right| \times 2 \frac{C \sqrt{T}}{1/4} \leq 24C T^{1/2-\eps} |\theta'-\theta| \leq 12C \left( 1+T^{1-2\eps}|\theta'-\theta|^2 \right). \]
By combining our estimates on \eqref{log_radon_nikodym}, we obtain, for every $\ell$ compatible with $E_C(\theta)$:
\[ \P \left( \mbox{$A$ observes $\ell$} |\Theta=\theta' \right) \geq \P \left( \mbox{$A$ observes $\ell$} |\Theta=\theta \right) \exp \left( -c-c T^{1-2\eps} |\theta'-\theta|^2 \right), \]
with $c=12C+96$. This proves \eqref{truncatd_continuity} and the lemma.
\end{proof}

The next lemma expresses the risk of collision: if $f^A_i(\theta)$ and $f^B_i(\theta')$ are both large for $\theta'$ close to $\theta$, then there is a risk that both players pull the arm $i$ and a large loss occurs. In all the rest of the paper, we will write $x \succeq y$ if $x$ is larger than $y$ times an absolute constant, which does not depend on $t$ or $T$ or $\theta$, but which may vary from line to line.

\begin{lemma}\label{lem_no_collision}
There is an absolute constant $\eta$ such that the following holds. Assume that there is an arm $i$ and $\theta, \theta'$ with $|\theta'-\theta| \leq \eta T^{\eps} \sqrt{\frac{\log T}{T}}$, such that
\[ f_i^A(\theta) \geq \frac{1}{10} \mbox{ and } f_i^B(\theta') \geq \frac{1}{10}. \]
Then $r_t(\theta) \succeq T^{-\eps/2}$ and $\E \left[ r_t(\Theta) \right] \succeq T^{-\frac{1}{2}+\frac{\eps}{2}}$.
\end{lemma}

\begin{proof}
Since every term in \eqref{decomposition_rt} is nonnegative, if $T$ is large enough so that all the $p_i$ are at most $\frac{3}{4}$, we can write
\begin{align*}
r_t(\theta) & \geq f_i^A(\theta) f_i^B(\theta) \left(2-p^*(\theta) \right)\\
& \geq \frac{1}{2} f_i^A(\theta) \left( f_i^B(\theta')-\frac{1}{20} \right) \exp \left( -c- c |\theta'-\theta|^2 T^{1-2\eps} \right)\\
& \succeq T^{-c \eta^2}\\
& \succeq T^{-\eps/2},
\end{align*}
provided $\eta$ was chosen small enough compared to $\eps$. The second inequality uses Lemma \ref{lem_regularity_f} with $\delta=\frac{1}{20}$. For the second point of the lemma, assume without loss of generality $\theta<\theta'$. For every $\theta''$ in the interval
\begin{equation}\label{eqn_interval_collision}
\left[ \theta-T^{\eps-1/2}, \theta'+T^{\eps-1/2} \right],
\end{equation}
we have $|\theta''-\theta|, |\theta''-\theta'| \leq 2 \eta T^{\eps} \sqrt{\frac{\log T}{T}}$ (provided $T$ is large enough), so Lemma \ref{lem_regularity_f} gives
\[ f_i^A(\theta'') \succeq T^{-4\eta^2c} \succeq T^{-\eps/4} \mbox{ and } f_i^B(\theta'') \succeq T^{-4\eta^2c} \succeq T^{-\eps/4}\]
provided $\eta$ is small enough. Hence $r_t(\theta'') \succeq T^{-\eps/2}$. Moreover, we know from \eqref{decomposition_rt} that $r_t(\Theta) \geq 0$, so
\[\E \left[ r_t(\Theta) \right] \succeq T^{-\eps/2} \P \left( \theta-T^{\eps-1/2} \leq \Theta \leq \theta'+T^{\eps-1/2} \right) \geq T^{-\eps/2} \times \frac{2 T^{\eps-1/2}}{4\pi} \succeq T^{-1/2+\eps/2},\]
where in the end we used the law of $\Theta$ \eqref{law_theta}.
\end{proof}

\begin{remark}
This is the only place in the proof where it was necessary that the fluctuations of $(p_1, p_2, p_3)$ are of order $T^{-\eps}$ instead of $1$. If the fluctuations were constant, the interval of \eqref{eqn_interval_collision} would have size $T^{-1/2}$ instead of $T^{\eps-1/2}$.
\end{remark}

We now define several regions on the unit circle. Our goal will then be to show in a quantitative way that the players must make certain choices on each of these regions (Lemmas \ref{lem_pick_best_arm} and \ref{lem_somewhere_best_two_arms}). More precisely, we write:

\begin{itemize}
\item[$\bullet$]
$I_1=\left[ \frac{\pi}{3}-2T^{-1/2+2\eps}, \frac{\pi}{3}+2T^{-1/2+2\eps} \right]$,
\item[$\bullet$]
$I_2=\left[ \pi-2T^{-1/2+2\eps}, \pi+2T^{-1/2+2\eps} \right]$,
\item[$\bullet$]
$I_3=\left[ \frac{5\pi}{3}-2T^{-1/2+2\eps}, \frac{5\pi}{3}+2T^{-1/2+2\eps} \right]$,
\item[$\bullet$]
$I_{12}=\left[ \frac{\pi}{3}+T^{-1/2+2\eps}, \pi-T^{-1/2+2\eps} \right]$,
\item[$\bullet$]
$I_{23}=\left[ \pi+T^{-1/2+2\eps}, \frac{5\pi}{3}-T^{-1/2+2\eps} \right]$,
\item[$\bullet$]
$I_{31}=\left[ \frac{5\pi}{3}+T^{-1/2+2\eps}, 2\pi \right] \cup \left[ 0, \frac{\pi}{3}-T^{-1/2+2\eps} \right]$.
\end{itemize}
See also Figure \ref{fig_sets_I12_etc} to see what these intervals look like. Basically, $I_i$ is the region where the arm $i$ is way better than the two others but the two others are close to each other. $I_{i_1 i_2}$ is the region where the arms $i_1$ and $i_2$ are significantly better than the last one. Note also that $I_1 \cup I_2 \cup I_3$ is precisely the set $I$ of \eqref{law_theta} where the distribution of $\Theta$ is "reinforced".

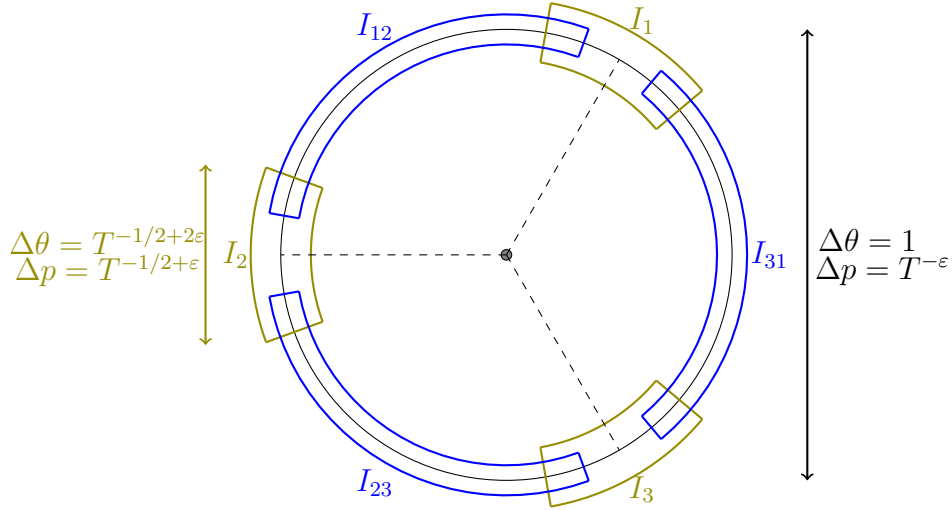
\begin{figure}
\begin{center}
\begin{tikzpicture}
\draw(0,0)circle(3cm);
\draw [olive,thick,domain=40:80] plot ({2.6*cos(\x)}, {2.6*sin(\x)});
\draw [olive,thick,domain=40:80] plot ({3.4*cos(\x)}, {3.4*sin(\x)});
\draw[olive, thick](40:2.6)--(40:3.4);
\draw[olive, thick](80:2.6)--(80:3.4);
\draw [olive,thick,domain=160:200] plot ({2.6*cos(\x)}, {2.6*sin(\x)});
\draw [olive,thick,domain=160:200] plot ({3.4*cos(\x)}, {3.4*sin(\x)});
\draw[olive, thick](160:2.6)--(160:3.4);
\draw[olive, thick](200:2.6)--(200:3.4);
\draw [olive,thick,domain=280:320] plot ({2.6*cos(\x)}, {2.6*sin(\x)});
\draw [olive,thick,domain=280:320] plot ({3.4*cos(\x)}, {3.4*sin(\x)});
\draw[olive, thick](280:2.6)--(280:3.4);
\draw[olive, thick](320:2.6)--(320:3.4);

\draw [blue,thick,domain=70:170] plot ({2.8*cos(\x)}, {2.8*sin(\x)});
\draw [blue,thick,domain=70:170] plot ({3.2*cos(\x)}, {3.2*sin(\x)});
\draw[blue, thick](70:2.8)--(70:3.2);
\draw[blue, thick](170:2.8)--(170:3.2);
\draw [blue,thick,domain=190:290] plot ({2.8*cos(\x)}, {2.8*sin(\x)});
\draw [blue,thick,domain=190:290] plot ({3.2*cos(\x)}, {3.2*sin(\x)});
\draw[blue, thick](190:2.8)--(190:3.2);
\draw[blue, thick](290:2.8)--(290:3.2);
\draw [blue,thick,domain=310:410] plot ({2.8*cos(\x)}, {2.8*sin(\x)});
\draw [blue,thick,domain=310:410] plot ({3.2*cos(\x)}, {3.2*sin(\x)});
\draw[blue, thick](310:2.8)--(310:3.2);
\draw[blue, thick](410:2.8)--(410:3.2);

\draw(0,0)node{};
\draw[dashed](0,0)--(60:3);
\draw[dashed](0,0)--(180:3);
\draw[dashed](0,0)--(300:3);

\draw[<->, olive, thick](-4,-1.2)--(-4,1.2);
\draw[olive](-5.3,0.2)node[texte]{$\Delta \theta = T^{-1/2+2\eps}$};
\draw[olive](-5.3,-0.2)node[texte]{$\Delta p = T^{-1/2+\eps}$};

\draw[<->, thick](4,-3)--(4,3);
\draw(4.8,0.2)node[texte]{$\Delta \theta = 1$};
\draw(5,-0.2)node[texte]{$\Delta p = T^{-\eps}$};

\draw[olive](180:3.6)node[texte]{$I_2$};
\draw[olive](60:3.6)node[texte]{$I_1$};
\draw[olive](300:3.6)node[texte]{$I_3$};
\draw[blue](120:3.5)node[texte]{$I_{12}$};
\draw[blue](0:3.5)node[texte]{$I_{31}$};
\draw[blue](240:3.5)node[texte]{$I_{23}$};
\end{tikzpicture}
\end{center}
\caption{The sets $I_i$ and $I_{i_1 i_2}$.}\label{fig_sets_I12_etc}
\end{figure}

The next lemma means that in the interval $I_i$, it is absolutely necessary that one of the players picks the arm $i$.

\begin{lemma}\label{lem_pick_best_arm}
Let $i_1,i_2,i_3$ be any permutation of the indices $1,2,3$. Assume that there is $\theta \in I_{i_1}$ such that
\[f_{i_2}^A(\theta) f_{i_3}^B(\theta) \geq \frac{1}{100}. \]
Then $r_t(\theta) \succeq T^{-\eps}$ and $\E \left[ r_t(\Theta) \right] \succeq T^{-2\eps}$.
\end{lemma} 

\begin{proof}
Without loss of generality, assume $i_1=1, i_2=2, i_3=3$. Since each term in \eqref{decomposition_rt} is nonnegative, we have
\[r_t(\theta) \geq f_{2}^A(\theta) f_{3}^B(\theta) \left( p_2(\theta)+p_3(\theta)-p^*(\theta) \right) \geq \frac{1}{100} \left( \max(p_2(\theta), p_3(\theta)) -p_1(\theta) \right) \succeq T^{-\eps}, \]
by the definition of $I_1$.

Similarly, for every $\theta'$ with $|\theta'-\theta| \leq T^{\eps-1/2}$, by Lemma \ref{lem_regularity_f}, we have $r_t(\theta') \succeq r_t(\theta) \succeq T^{-\eps}$. Therefore:
\[ \E \left[ r_t(\Theta) \right] \succeq T^{-\eps} \P \left( |\Theta-\theta| \leq T^{\eps-1/2} \right) \succeq T^{-\eps} \frac{T^{\eps-1/2}}{T^{2\eps-1/2}},\]
where the last inequality follows from the law of $\Theta$ \eqref{law_theta}, and more precisely the fact that it is "reinforced" on $I_1 \cup I_2 \cup I_3$.
\end{proof}

After Lemmas \ref{lem_no_collision} and \ref{lem_pick_best_arm}, we now state a third constraint on the strategy of the players. This one states that a suboptimal choice cannot be made on a too large region, and in particular not on the whole region $I_1 \cap I_{12}$.

\begin{lemma}\label{lem_somewhere_best_two_arms}
Let $i_1,i_2,i_3$ be any permutation of the indices $1,2,3$.
\begin{itemize}
\item[$\bullet$]
Let $\theta \in I_{i_1} \cap I_{i_1 i_2}$. If
\[f_{i_1}^A(\theta)f_{i_3}^B(\theta) \geq \frac{1}{100} \mbox{ or } f_{i_3}^A(\theta)f_{i_1}^B(\theta) \geq \frac{1}{100},\]
then $r_t(\theta) \succeq T^{-1/2+\eps}$.
\item[$\bullet$]
If
\[f_{i_1}^A(\theta)f_{i_3}^B(\theta) + f_{i_3}^A(\theta)f_{i_1}^B(\theta) \geq \frac{2}{100}\]
for all $\theta \in I_{i_1} \cap I_{i_1 i_2}$, then $\E \left[ r_t(\Theta) \right] \succeq T^{-1/2+\eps}$.
\end{itemize}
\end{lemma}

\begin{proof}
Without loss of generality, assume $i_1=1, i_2=2, i_3=3$, so that $p_1(\theta)<p_2(\theta)<p_3(\theta)$ on $I_1 \cap I_{12}$. For the first point, by \eqref{decomposition_rt}, we have
\[r_t(\theta) \geq f_1^A(\theta) f_3^B(\theta) \left( p_3(\theta)-p_2(\theta)\right) \succeq T^{-1/2+\eps}, \]
where the last inequality follows from the definition of $I_{12}$.

This implies that under the assumptions of the second point, we have $r_t(\theta) \succeq T^{-1/2+\eps}$ for all $\theta \in I_1 \cap I_{12}$, so
\[ \E \left[r_t(\Theta) \right] \succeq T^{-1/2+\eps} \P \left( \Theta \in I_{12} \cap I_1 \right) \succeq T^{-1/2+\eps} \]
by \eqref{law_theta}.
\end{proof}

\subsection{Proof of Theorem \ref{thm:lower}}
We recall that $1 \leq t \leq T$ is fixed. As noted earlier, it is sufficient to check $\E \left[ r_t(\Theta) \right] \succeq \sqrt{\frac{\log T}{T}}$. For each $\theta$, let $a(\theta)$ (resp. $b(\theta)$) be the set of arms $i$ such that $f_i^A(\theta)$ (resp. $f_i^B(\theta)$) is at least $\frac{1}{10}$.

It follows from Lemma \ref{lem_no_collision} that if $\E \left[ r_t(\Theta) \right] \preceq \sqrt{\frac{\log T}{T}}$, then $a(\theta) \cap b(\theta) = \emptyset$ and clearly $a(\theta)$ and $b(\theta)$ are nonempty, so only the following situations can occur:
\begin{itemize}
\item
$a(\theta)$ and $b(\theta)$ are disjoint singletons;
\item
$a(\theta)$ is a singleton and $b(\theta)$ its complement;
\item
$b(\theta)$ is a singleton and $a(\theta)$ its complement.
\end{itemize}
We denote by $E(\theta)$ the pair $(a(\theta), b(\theta))$. We will write $E(\theta)$ in a compact form. For example, if $a(\theta)=\{1,3\}$ and $b(\theta)=\{2\}$, we will write $E(\theta)=(13,2)$. The $12$ possible values of $E(\theta)$ split the circle on which $\theta$ lives into regions. Since the functions $f_i^A$ and $f_i^B$ are analytic by Lemma \ref{lem_regularity_f}, these regions are finite unions of intervals. Moreover, Lemma \ref{lem_no_collision} shows that $a(\theta) \cap b(\theta') = \emptyset$ if $|\theta'-\theta| \leq \eta T^{\eps} \sqrt{\frac{\log T}{T}}$, so certain regions may not touch each other. More precisly, the graph of possible adjacence of these regions is summed up on Figure \ref{fig_graph_pairs_collisions}.

\begin{figure}
\begin{center}
\begin{tikzpicture}
\draw(0:2.5)--(30:2.5);
\draw(30:2.5)--(60:2.5);
\draw(60:2.5)--(90:2.5);
\draw(90:2.5)--(120:2.5);
\draw(120:2.5)--(150:2.5);
\draw(150:2.5)--(180:2.5);
\draw(180:2.5)--(210:2.5);
\draw(210:2.5)--(240:2.5);
\draw(240:2.5)--(270:2.5);
\draw(270:2.5)--(300:2.5);
\draw(300:2.5)--(330:2.5);
\draw(330:2.5)--(0:2.5);

\draw[dashed](0:2.5)to[bend left](60:2.5);
\draw[dashed](60:2.5)to[bend left](120:2.5);
\draw[dashed](120:2.5)to[bend left](180:2.5);
\draw[dashed](180:2.5)to[bend left](240:2.5);
\draw[dashed](240:2.5)to[bend left](300:2.5);
\draw[dashed](300:2.5)to[bend left](0:2.5);


\draw(0:2.5)node[white]{$1,2$};
\draw(30:2.5)node[white]{$1,23$};
\draw(60:2.5)node[white]{$1,3$};
\draw(90:2.5)node[white]{$12,3$};
\draw(120:2.5)node[white]{$2,3$};
\draw(150:2.5)node[white]{$2,13$};
\draw(180:2.5)node[white]{$2,1$};
\draw(210:2.5)node[white]{$23,1$};
\draw(240:2.5)node[white]{$3,1$};
\draw(270:2.5)node[white]{$3,12$};
\draw(300:2.5)node[white]{$3,2$};
\draw(330:2.5)node[white]{$13,2$};
\end{tikzpicture}
\end{center}
\caption{The collision graph: the vertices are the possible values of $E(\theta)$. The pairs of vertices linked by a full edge correspond to regions that may be neighbour. Note that $(1,2)$ and $(1,3)$ are not linked by a full edge, because at the boundary we would have $f_2^B(\theta)=f_3^B(\theta)=\frac{1}{10}$ but $f_1^B(\theta)<\frac{1}{10}$, which is not possible since $f_1^B+f_2^B+f_3^B=1$. The vertices not linked by any edge correspond to regions which must be separated by at least $\eta T^{\eps} \sqrt{\frac{\log T}{T}}$ to avoid the risk of a collision.}\label{fig_graph_pairs_collisions}
\end{figure}
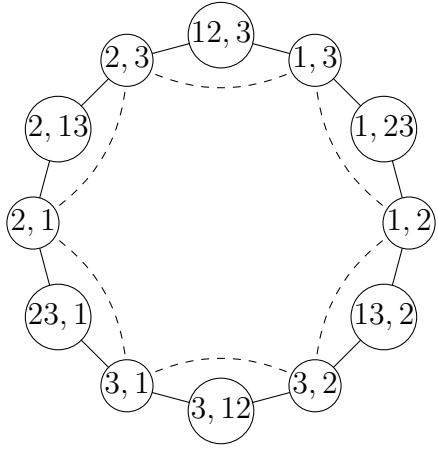

Moreover, if $\E \left[ r_t(\Theta) \right] \preceq \sqrt{\frac{\log T}{T}}$, then Lemmas \ref{lem_pick_best_arm} and \ref{lem_somewhere_best_two_arms} imply respectively the following.
\begin{enumerate}
\item\label{item_1}
For $i \in \{1,2,3\}$ and $\theta \in I_i$, we have $i \in a(\theta) \cup b(\theta)$;
\item\label{item_2}
for any permutation $i_1, i_2, i_3$ of the indices $1,2,3$, there is $\theta_{i_1 i_2} \in I_{i_1} \cap I_{i_1 i_2}$ such that $\{i_1, i_3\}$ is not included in $a(\theta_{i_1 i_2}) \cup b(\theta_{i_1 i_2})$. Since $i_1$ is always in the union by the previous item, it means that $E(\theta_{i_1 i_2})$ has to be $(i_1,i_2)$ or $(i_2,i_1)$.
\end{enumerate}

\begin{lemma}\label{lem_cyclic_argument}
There is a permutation $i_1, i_2, i_3$ of the indices $1,2,3$ such that:
\[E(\theta_{i_1 i_2})=(i_1,i_2) \mbox{ but } E(\theta_{i_2 i_1})=(i_2,i_1), \mbox{ or } E(\theta_{i_1 i_2})=(i_2,i_1) \mbox{ but } E(\theta_{i_2 i_1})=(i_1,i_2).\]
\end{lemma}

\begin{proof}
Suppose this is not the case, and assume without loss of generality that $E(\theta_{12})=(1,2)$. By Item \ref{item_1} above, we know that for every $\theta \in I_1$, the arm $1$ must be in exactly one of the two sets $a(\theta)$ and $b(\theta)$. Since $I_1$ is connected, it is always in the same set, so $1 \in a(\theta)$. In particular, since $\theta_{13} \in I_1$, we have $1 \in a(\theta_{13})$, so $E(\theta_{13})=(1,3)$.

But by our assumption that we are on a counter-example to Lemma \ref{lem_cyclic_argument}, it follows that $E(\theta_{31})=(1,3)$. By the same argument using Item \ref{item_1}, this implies $E(\theta_{32})=(2,3)$, so by our assumption $E(\theta_{23})=(2,3)$. Hence $E(\theta_{21})=(2,1)$ by Item \ref{item_1} and finally $E(\theta_{12})=(2,1)$ by our assumption. This is a contradiction.
\end{proof}

We are now in position to conclude the proof of Theorem \ref{thm:lower}. We consider a counter-example to \eqref{regret_at_one_step}. By Lemma \ref{lem_cyclic_argument}, without loss of generality, we can assume $E(\theta_{12})=(1,2)$ and $E(\theta_{21})=(2,1)$, where $\theta_{12} \in I_1 \cap I_{12}$ and $\theta_{21} \in I_2 \cap I_{12}$, so $\theta_{12}<\theta_{21}$. We define
\[\widehat{\theta}=\inf \{ \theta \in [\theta_{12}, \theta_{21}] | E(\theta)=(2,1) \},\]
\[ \widetilde{\theta}=\sup \{ \theta \in [\theta_{12}, \widehat{\theta}] | E(\theta)=(1,2)\}. \]
We note that by definition of $I_{12}$, we have
\[ \frac{\pi}{3}+T^{-1/2+2\eps} \leq \widetilde{\theta} < \widehat{\theta} \leq \pi-T^{-1/2+2\eps},\]
with $\widehat{\theta}-\widetilde{\theta} \geq \eta T^{\eps}\sqrt{\frac{\log T}{T}}$ to avoid collisions (see Figure \ref{fig_graph_pairs_collisions}). By definition, for $\widetilde{\theta}<\theta<\widehat{\theta}$, we have $E(\theta) \ne (1,2), (2,1)$, so $3 \in a(\theta) \cup b(\theta)$. But note that on Figure \ref{fig_graph_pairs_collisions}, the vertices $(1,2)$ and $(2,1)$ disconnect the graph into two parts: the "top" part, where $3 \in b(\theta)$, and the "bottom" part, where $3 \in a(\theta)$. It follows that either $3 \in a(\theta)$ for all $\widetilde{\theta}<\theta<\widehat{\theta}$, or $3 \in b(\theta)$ for all such $\theta$. Without loss of generality, assume that we are in the first case.

To finish the proof, we distinguish three cases according to the values of $\widetilde{\theta}$ and $\widehat{\theta}$ in the interval $I_{12}$.
\begin{itemize}
\item[$\bullet$]
Case 1: $\frac{\pi}{3}+\frac{\pi}{100} \leq \widetilde{\theta}<\widehat{\theta}$.\\
In this case, note that by the graph of Figure \ref{fig_graph_pairs_collisions}, the region where $E(\theta)=(3,2)$ must be separated from $\widehat{\theta}$ by at least $\eta T^{\eps}\sqrt{\frac{\log T}{T}}$. Hence, there is an interval $J$ of length at least $\eta T^{\eps}\sqrt{\frac{\log T}{T}}$ where $3 \in a(\theta)$ and $1 \in b(\theta)$. For any $\theta$ in this interval, we have
\begin{align*}
r_t(\theta) & \geq f_3^A(\theta) f_1^B(\theta) \left( p_1(\theta)+p_3(\theta)-p^*(\theta) \right)\\
& \geq \frac{1}{100} \left( p_3(\theta)-p_2(\theta) \right)\\
& \succeq T^{-\eps},
\end{align*}
where the second inequality follows from the definitions of $a(\theta)$ and $b(\theta)$, and the last one from $\theta>\frac{\pi}{3}+\frac{\pi}{100}$. From the law of $\Theta$, it follows that
\[ \E \left[ r_t(\Theta) \right] \succeq T^{-\eps} \P \left( \Theta \in J \right) \geq T^{-\eps} \times \frac{1}{4\pi} \eta T^{\eps}\sqrt{\frac{\log T}{T}} \succeq \sqrt{\frac{\log T}{T}}.\]
\item[$\bullet$]
Case 2: $\widetilde{\theta} < \widehat{\theta} \leq \pi-\frac{\pi}{100}$.\\
This case is similar to the first one where we exchange the roles of the arms $1$ and $2$: there is an interval $J' \subset \left[ \frac{\pi}{3}, \pi-\frac{\pi}{100} \right]$ with length at least $\eta T^{\eps}\sqrt{\frac{\log T}{T}}$ where $3 \in a(\theta)$ and $2 \in b(\theta)$. On this interval, we have
\[r_t(\theta) \geq f_3^A(\theta) f_2^B(\theta) \left( p_3(\theta)-p_1(\theta) \right) \succeq T^{-\eps},\]
so we get $\E \left[ r_t(\Theta) \right] \succeq \sqrt{\frac{\log T}{T}}$.
\item[$\bullet$]
Case 3: $\widetilde{\theta}<\frac{\pi}{3}+\frac{\pi}{100}<\pi-\frac{\pi}{100}<\widehat{\theta}$.\\
In this case, we have $3 \in a(\theta)$ on the full interval $\left[ \frac{\pi}{3}+\frac{\pi}{100}, \pi-\frac{\pi}{100}\right]$, so for any $\theta$ in that interval we have
\[r_t(\theta) \geq f_3^A(\theta) \left( p_3(\theta)-\max \left( p_1(\theta), p_2(\theta) \right) \right) \succeq T^{-\eps}.\]
Since this interval is macroscopic, the variable $\Theta$ lands in it with probability $\succeq 1$, so $\E \left[ r_t(\Theta) \right] \succeq T^{-\eps}$, which concludes the proof.
\end{itemize}

\begin{remark}
Separating different cases was necessary in the end: for example, if the interval $[\widetilde{\theta}, \widehat{\theta}]$ is very close to $\frac{\pi}{3}$, then the arm $2$ is barely better than $3$, so we lose almost nothing on the interval where $E(\theta)=(3,1)$. However, we lose a lot when $E(\theta)=(3,2)$.
\end{remark}

\section{Bandit upper bound with collisions} \label{sec:add}
We prove here Theorem \ref{thm:add}.

\subsection{Strategy}
We denote by $q_t^X(i)$ the empirical mean for arm $i$ using the observed rewards on arm $i$ up to time $t$ by player $X$.
\newline

Initialization phase: During $t_0=40 \sqrt{T \log(T)}$ rounds, Alice stays on action $3$, and Bob alternates between action $1$ and $2$. 
We set 
\[
\mathcal{B} = \{1,2,3\} \setminus \left\{ i \in \{1,2\} : q_{t_0}^B(i) \geq 1 - \sqrt{\frac{\log(T)}{T}} \right\} \,.
\] 
to be the set of {\em valid} actions for Bob. Note that this set is always nonempty since $3 \in \mathcal{B}$.

For the rest of the game ($t > t_0$), Alice and Bob will play a phase-based strategy. Assume that $t_0 = 2^{k_0}$ for some $k_0 \in \N$. For each phase $k \geq k_0$ we define the strategies as follows.
\newline

\textbf{Alice:} 
\begin{enumerate}[label=(\roman*)]
\item \label{item_i}
At the beginning of a phase, i.e. for $t = 2^{k}+1$, if $|q_t^A(1) - q_t^A(2)| \geq 10 \sqrt{\frac{\log(T)}{t-t_0}}$ then for the rest of the game Alice stays put on $\arg\max_{i \in \{1,2\}} q_t^A(i)$. 
\item \label{item_ii}
Otherwise for all $t \in [2^k+1, 2^{k+1}]$ Alice plays action $1$ if $k$ is odd and action $2$ if $k$ is even.
\end{enumerate}

\textbf{Bob:}
\begin{enumerate}[label=(\roman*)]\setcounter{enumi}{2}
\item \label{item_iii}
If $\mathcal{B} = \{1,2,3\}$, then at the beginning of a phase $t = 2^{k}+1$, if $q_t^B(3) - \max_{i \in \{1,2\}} q_t^B(i) \geq 100 \sqrt{\frac{\log(T)}{t}}$, then we remove action $3$ from the valid actions for Bob, i.e., we set $\mathcal{B} = \{1,2\}$.
\item \label{item_iv}
On the other hand, if $q_t^B(3) - \min_{i \in \{1,2\} \cap \mathcal{B}} q_t^B(i) \leq - 100 \sqrt{\frac{\log(T)}{t}}$, then Bob plays action $3$ until the end of the game.
\item \label{item_v}
For any $t \in [2^k+1, 2^{k+1}]$, Bob plays alternatively between the actions in $\mathcal{B} \setminus \{i_k\}$, where $i_k=1$ if $k$ is odd, and $i_k =2$ if $k$ is even. (Note that $\mathcal{B} \setminus \{i_k\}$ is always non-empty since per item \ref{item_iii} we either have $3 \in \mathcal{B}$ or $\mathcal{B} = \{1,2\}$.) 
\item \label{item_vi}
If for all $t \in [2^k+1, 2^k + 40 \sqrt{T \log(T)}]$ Bob observes only losses of $1$ on action $3-i_k$ then it stops the current phase. Bob will now restart the game, forget all its previous observations (including the initialization phase) and play some single-player multi-armed bandit strategy on the set of actions $\{1,2, 3\} \setminus \{3-i_k\}$.
\end{enumerate} 

\subsection{Collisions analysis}
First we control the number of collisions. Let us denote $\tau^A$ for the time at which Alice fixates on an action according to item \ref{item_i}, and $\tau^B$ for the time at which Bob restarts according to item \ref{item_vi}. The next result means that the times $\tau^A$ and $\tau^B$ behave "as one would expect".
\begin{lemma} \label{lem:addnocoll}
With probability at least $1-\Omega(1/T)$, either Bob stops the game and play action $3$ forever (item \ref{item_iv}), or one has $\tau^B = \tau^A + 40 \sqrt{T \log(T)}$, and moreover in that case Bob restarts with the action set $\{1,2,3\}$ minus the action where Alice is fixated.
\end{lemma}
Note that this result implies that with high probability the total number of collisions is smaller than $40 \sqrt{T \log(T)}$. Indeed, by construction there can be no collisions before $\min(\tau^A, \tau^B)$, and moreover there are no collisions after $\max(\tau^A, \tau^B)$ if Bob's actions after restart do not include Alice's fixated action.

\begin{proof}
Using Bernstein's inequality, one has with probability at least $1-1/T^2$ that for $i \in \{1,2\}$,
\begin{equation} \label{firstevent}
p_i \geq 1 - \frac{1}{10} \sqrt{\frac{\log(T)}{T}} \Rightarrow i \not\in \mathcal{B} \,.
\end{equation}
Consider now an action $i$ with $p_i \leq 1 - \frac{1}{10} \sqrt{\frac{\log(T)}{T}}$, and say that Bob plays action $i$ for $20 \sqrt{T \log(T)}$ rounds while Alice is not playing it. Then, using that $\P(\mathrm{Bin}(k, 1-\epsilon) = k) = (1-\epsilon)^k \leq \exp( - k \epsilon)$, the probability that he sees only losses of $1$ during those rounds is smaller than $1/T^2$.
\newline

Observe that by construction, Alice and Bob cannot collide before $\min(\tau^A, \tau^B)$. In particular, if Alice is not fixated, then Bob can only restart using item \ref{item_vi} if he sees all $1$'s for $20 \sqrt{T \log(T)}$ trials of a valid action while Alice is there. The above statements show, via a simple union bound, that the probability of this happening at some phase is $O(1/T)$. Thus we see that with high probability, if Bob restarts, then Alice must be fixated. The converse is also clearly true, namely once Alice fixates, Bob will restart within the next phase by construction, and moreover he will restart without including Alice's action.
\end{proof}

\subsection{Proof of Theorem \ref{thm:add}}
We first note that the contribution of the initialization phase to the regret is $O(\sqrt{T \log T})$, so we can focus on the part $t>t_0$.
All statements in this proof will hold with probability $1-\Omega(1/T)$ (in particular we assume that the event of Lemma \ref{lem:addnocoll} holds true, so that there are no collisions before Alice fixates). Without loss of generality we assume that $p_1 \leq p_2$. Alice always fixates on action $1$ by item \ref{item_i} after at most $O(\log(T)/(p_2 - p_1)^2)$ rounds, and in particular her regret with respect to playing action $1$ all the time is at most $O \left( \min(\log(T)/(p_2-p_1), T (p_2 - p_1) \right)$ (since she always plays either action $1$ or $2$). Thus her regret with respect to action $1$ is always $O(\sqrt{T \log(T)})$. Moreover, after $\tau^B$, the regret of Bob is always controlled too by the same argument. Thus we only need to analyze the regret of Bob before $\tau^B$ (in fact before $\tau^A$, since he must restart within $O(\sqrt{T \log(T)})$ rounds from $\tau^A$ by Lemma \ref{lem:addnocoll}). To do so, we distinguish three cases. In case $1$ and case $2$ the best action will be action $3$, so that we need to calculate the regret of Bob with respect to playing action $3$ all the time. In those first two cases we will see that it can only help if Bob removes an arm at initialization. On the other hand in case $3$, action $1$ is the best and action $2$ the second best, so that we need to calculate the regret of Bob with respect to playing action $2$ all the time (since the regret of Alice is calculated with respect to action $1$). In that latter case we need to argue about what happens if an action is removed at the end of initialization.

\paragraph{Case 1: $p_3 \leq p_1 \leq p_2$.} We know that $\tau^A \leq O(\log(T)/(p_2-p_1)^2)$, so Bob plays action $1$ and $2$ at most that many times (before $\tau^A$). But we also know by item \ref{item_iv} that Bob will stop the game and play action $3$ after at most $O(\log(T) / (p_1-p_3)^2)$ rounds (this is true even if an arm is removed at initialization). Thus in fact he plays action $1$ and $2$ (before $\tau^A$) at most $O\left( \frac{\log(T)}{\max(p_1 - p_3, p_2-p_1)^2} \right)$, and the resulting regret is at most $O\left( \log(T) \frac{p_2-p_3}{\max(p_1 - p_3, p_2-p_1)^2} \right) = O\left(\frac{\log(T)}{p_2 - p_3} \right)$. Note also that Bob's regret is always smaller than $T (p_2 - p_3)$. Thus in this case we get again $O(\sqrt{T \log(T)})$.

\paragraph{Case 2: $p_1 \leq p_3 \leq p_2$.} Here Bob simply pays less regret than Alice before $\tau^A$ (since any suboptimal play of Alice cost $p_2-p_1$ while a suboptimal play of Bob only cost $p_2 - p_3$).

\paragraph{Case 3: $p_1 \leq p_2 \leq p_3$.} First note that the condition of item \ref{item_iv} where Bob stops the game to play action $3$ forever will not trigger. 

Next let us assume that Bob does not remove any arm at the end of initialization. Then Bob will remove action $3$ from its valid actions using item \ref{item_iii} after at most $O(\log(T)/(p_3 - p_2)^2)$ rounds, so the regret suffered by this is at most $O(\log(T)/(p_3 - p_2))$. Note also that Bob's regret is always smaller than $T (p_3 - p_2)$. Thus in this case we get again $O(\sqrt{T \log(T)})$. 

Finally let us assume that Bob does remove either arm $1$ or arm $2$ at the end of initialization. Then it must be that $p_3 - p_2 = O \left( \sqrt{\frac{\log(T)}{T}} \right)$ (by \eqref{firstevent} and the fact that $p_3 \geq p_2$), and thus the regret is automatically $O(\sqrt{T \log(T)})$.

\subsubsection*{Acknowledgment}
We are grateful to Omer Angel for a suggestion that simplified the proof of Theorem \ref{thm:lower}. We also thank three anonymous referees for useful remarks.

\bibliographystyle{plainnat}
\bibliography{newbib}

\end{document}